\documentclass[11pt,a4paper,final]{article}

\PassOptionsToPackage{numbers, compress}{natbib}
\usepackage[left=1in,right=1in,top=1in,bottom=1in]{geometry}

\usepackage[utf8]{inputenc} % allow utf-8 input
\usepackage[T1]{fontenc}    % use 8-bit T1 fonts
\usepackage{hyperref}       % hyperlinks
\usepackage{url}            % simple URL typesetting
\usepackage{booktabs}       % professional-quality tables
\usepackage{amsfonts}       % blackboard math symbols
\usepackage{nicefrac}       % compact symbols for 1/2, etc.
\usepackage{microtype}      % microtypography
\usepackage{xcolor}         % colors

\usepackage[ruled]{algorithm2e} % For algorithms
\usepackage{tikz}
\usetikzlibrary{arrows}
\usetikzlibrary{shapes.multipart}
\usepackage{caption}

\SetAlFnt{\small}
\SetAlCapFnt{\small}
\SetAlCapNameFnt{\small}
\SetAlCapHSkip{0pt}
\IncMargin{-\parindent}
\usepackage{natbib}

\newcommand{\declarecolor}[2]{\definecolor{#1}{RGB}{#2}\expandafter\newcommand\csname #1\endcsname[1]{\textcolor{#1}{##1}}}

\declarecolor{White}{255, 255, 255}
\declarecolor{Black}{0, 0, 0}
\declarecolor{Maroon}{128, 0, 0}
\declarecolor{Coral}{255, 127, 80}
\declarecolor{Red}{182, 21, 21}
\declarecolor{LimeGreen}{50, 205, 50}
\declarecolor{DarkGreen}{0, 90, 0}
\declarecolor{Navy}{0, 0, 128}

\definecolor{plotblue}{HTML}{377eb8}
\definecolor{plotorange}{HTML}{ff7f00}
\definecolor{plotgreen}{HTML}{4daf4a}

\hypersetup{
    colorlinks=true,
    citecolor=DarkGreen,
    linkcolor=Navy,
    filecolor=magenta,      
    urlcolor=black,
    pdfpagemode=FullScreen,
}

\usepackage{mathtools}
\usepackage{amsthm}
\usepackage{bbm}
\usepackage{tikz} 
\usepackage{wrapfig}
\usepackage{floatrow}
\usepackage{enumitem}
\usepackage{authblk}
\usepackage{xparse}

\usepackage[suppress]{color-edits}
\addauthor{kh}{red}
\addauthor{aa}{blue}
\addauthor{sw}{blue}
\addauthor{cp}{green!60!black}

\usepackage{caption}
\usepackage{subcaption}
\usepackage{cleveref}

\newtheorem{theorem}{Theorem}[section]
\newtheorem{lemma}[theorem]{Lemma}

\newtheorem{corollary}[theorem]{Corollary}

\newtheorem{condition}[theorem]{Condition}
\newtheorem{assumption}[theorem]{Assumption}
\newtheorem{example}[theorem]{Example}
\newtheorem{definition}[theorem]{Definition}

\newcommand{\snr}{\mathsf{snr}}

\newcommand{\cN}{\mathcal{N}}
\newcommand{\cB}{\mathcal{B}}

\newcommand{\cV}{\mathcal{V}}
\newcommand{\cU}{\mathcal{U}}

\newcommand{\cY}{\mathcal{Y}}

\newcommand{\vr}{\mathbf{r}}
\newcommand{\vu}{\mathbf{u}}

\newcommand{\vy}{\mathbf{y}}

\newcommand{\vbeta}{\boldsymbol{\beta}}
\NewDocumentCommand{\pv}{m e{_} m}{%
  #1\IfValueT{#2}{_{#2}}^{(#3)}%
}
\newcommand{\range}[1]{[\![#1]\!]}
\newcommand{\E}{\mathbb E}

\begin{document}
\title{Strategyproof Decision-Making in Panel Data Settings and Beyond\footnote{A version of this paper was published as an extended abstract in SIGMETRICS 2024 ($50$th ACM SIGMETRICS International Conference on Measurement and Modeling of Computer Systems).}}
%
% The \author macro works with any number of authors. There are two commands
% used to separate the names and addresses of multiple authors: \And and \AND.
%
% Using \And between authors leaves it to LaTeX to determine where to break the
% lines. Using \AND forces a line break at that point. So, if LaTeX puts 3 of 4
% authors names on the first line, and the last on the second line, try using
% \AND instead of \And before the third author name.

\author[1]{Keegan Harris}
\author[2]{Anish Agarwal\thanks{For part of this work, Anish was a postdoc at Amazon, Core AI.}}
\author[3,4]{Chara Podimata\thanks{For part of this work, Chara was a FODSI postdoc at UC Berkeley.}}
\author[1]{Zhiwei Steven Wu}

\affil[1]{Carnegie Mellon University}
\affil[2]{Columbia University}
\affil[3]{Massachusetts Institute of Technology}
\affil[4]{Archimedes/Athena RC}
\affil[ ]{\texttt {\{keeganh, zstevenwu\}@cmu.edu}}
\affil[ ]{\texttt{aa5194@columbia.edu, podimata@mit.edu}}
\date{}
\maketitle
\vspace*{-0.4in}

\pagenumbering{gobble}
=\begin{abstract}
We consider the problem of decision-making using \emph{panel data}, in which a decision-maker gets noisy, repeated measurements of multiple \emph{units} (or \emph{agents}). 
We consider the setup used in synthetic control methods, where there is a pre-intervention period when the principal observes the outcomes of each unit, after which the principal uses these observations to assign a treatment to each unit.
Unlike this classical setting, we permit the units generating the panel data to be strategic, i.e. units may modify their pre-intervention outcomes in order to receive a more desirable intervention. 
The principal's goal is to design a \emph{strategyproof} intervention policy, i.e. a policy that assigns units their \khedit{utility-maximizing} interventions despite their potential strategizing.
We first identify a necessary and sufficient condition under which a strategyproof intervention policy exists, and provide a strategyproof mechanism with a simple closed form when one does exist.
Along the way, we prove impossibility results for \emph{strategic multiclass classification}, a natural extension of the well-studied (binary) strategic classification problem to the multiclass setting, which may be of independent interest.
When there are two interventions, we establish that there always exists a strategyproof mechanism, and provide an algorithm for learning such a mechanism. 
For three or more interventions, we provide an algorithm for learning a strategyproof mechanism if there exists a sufficiently large gap in the principal's rewards between different interventions. 
Finally, we empirically evaluate our model using real-world panel data collected from product sales over 18 months. 
We find that our methods compare favorably to baselines which do not take strategic interactions into consideration, even in the presence of model misspecification.\looseness-1
\end{abstract}
\newpage
\tableofcontents
\newpage
\pagenumbering{arabic}

\section{Introduction}\label{sec:intro}
In \emph{panel data} (or \emph{longitudinal data}) settings, one observes repeated, noisy, measurements of a collection of \emph{units} over a period of time, during which the units undergo different interventions.
For example, units can be individuals, companies, or geographic locations, and interventions can %range from 
represent discounts, health therapies, or tax regulations.
%
%This%\cpcomment{unclear what this ``this'' refers to} 
%\cpedit{Collecting measurements of units over a period of time} 
%
This is a ubiquitous way to collect data, and, as a result, the analysis of panel data has a long history in econometrics and statistics.
A common goal in the literature is to analyze how a \emph{principal} (e.g., business platform, regulatory agency) can do ``counterfactual inference'', i.e., estimate what will happen to a unit if it undergoes a variety of possible interventions.
The ultimate goal of such counterfactual inference is to enable data-driven decision-making, where one does not just estimate statistical parameters of interest, but actually uses data to make better decisions. In medical domains, for example, the goal typically is not just estimating  health outcomes for patients under different health therapies, but also a policy that selects appropriate therapies for new patients.
However, the leap from counterfactual inference to data-driven decision-making comes with additional challenges: namely, when units know that they will be assigned disparate interventions based on their reported data, they have incentives to strategize with their reports.
Such strategic interactions in panel data settings are observed in practice. 
For example,~\citet{caro2010zara} observe that Zara store managers strategically misreported store inventory information to higher-ups in order to maximize sales at their local branch.\looseness-1

A running example we will use throughout this paper is that of an e-commerce platform that wishes to give one of several possible discounts (interventions) to a new user to maximize some future metric of interest, say, engagement levels. 
\khedit{In this example time-steps are days/weeks/months, units are users, and outcomes are engagement levels.\footnote{\khedit{Another motivating example is the incentives of store managers in inventory management, as discussed in~\citet{caro2010zara}. Here, time-steps are days/weeks/months, units are different products at a given store, outcomes are (reported) inventory levels, and interventions are different levels of discount.}}} 
Suppose the company uses historical data to build a model that estimates the ``counterfactual'' trajectory of engagement levels of a new user under different discount policies, based on their observed trajectory of engagement levels thus far.
If a user knew this were the case, then there is a clear incentive for them to strategically modify their engagement levels to receive a larger discount.
Such strategic manipulations in response to data-driven decision-making have
%also \cpcomment{``also'' here a bit weird since from the writeup it's unclear whether such manipulation has been observed in the e-commerce example} 
been observed in other domains such as lending \citep{homonoff2021does} and search engine optimization \citep{davis2006search}.
%
%Hence, we focus on two natural questions that arise as we transition from counterfactual inference to decision-making: 
In this paper, we focus on \emph{strategyproof} intervention policies, i.e., policies that assign the \khedit{utility-maximizing} treatment to the units despite them strategically altering their data. Concretely, we answer two questions:
\vspace{-3mm}

{\centerline \em {\bf Q1:} Is it possible to design %strategyproof
intervention policies %estimators %\cpcomment{What exactly is an estimator here? Above we have said that we don't want to do *just* estimation. Maybe it'll be better to introduce the concept of a ``policy'' and then talk about SP policies} 
that are robust to strategic modification of data by units to receive a more favorable intervention? 
We call such policies \emph{strategyproof}.}

\vspace{-3mm}

{\centerline \em {\bf Q2:} Can we leverage the structure typically present in panel data to derive computationally-efficient algorithms for learning strategyproof intervention policies?}
\vspace{2mm}

Towards answering both questions, and in line with the e-commerce example described above, we build upon the framework for counterfactual inference with panel data called \emph{synthetic interventions} \citep{SI}, which itself is a generalization of the canonical framework of \emph{synthetic control} \citep{abadie2003economic, abadie2010synthetic}.
In both settings, there is a notion of a ``pre-intervention'' time period when all units are under control (i.e., no intervention), followed by a ``post-intervention'' time period, when each unit undergoes \emph{exactly one} of many possible interventions (including control).
Synthetic control methods can be used to estimate the counterfactual outcome if a unit did not undergo an intervention, i.e., remained under control.
Given its simplicity, synthetic control has become a ubiquitous technique in econometrics to estimate counterfactuals under control in high-stakes domains including e-commerce, healthcare, and policy evaluation. 
Synthetic interventions is a generalization which allows one to estimate counterfactual outcomes not just under control, but also under intervention.
Despite being a relatively new development, the synthetic interventions framework has been used by several companies across therapeutics, ride-sharing, and e-commerce, to estimate the best intervention for a given unit. 
Given the growing popularity of the synthetic control and synthetic interventions frameworks in high-stakes decision-making, an issue that will likely arise is the strategic manipulation of observed outcomes by units. 
In particular, just as in the e-commerce example, units have an incentive to manipulate their pre-intervention outcomes in order to get a more favorable intervention during the post-intervention period. 
Indeed, if such strategic manipulations are not taken into account, we establish that the synthetic interventions estimator can perform poorly---see~\Cref{sec:not-sp} for a formal example of such a scenario and~\Cref{sec:experiments} for an empirical example based on real-world e-commerce data. 
The goal of the principal in such strategic settings is to assign the ``correct'' intervention to each unit---the intervention which maximizes some objective function such as total user engagement---during the post-intervention period, despite possible strategic manipulations to the unit's pre-treatment behavior. 
As we observe, this can be thought of as designing a mechanism for {\em strategyproof multi-class classification} with panel data.
Although many of our results apply more broadly to strategyproof multi-class classification, we focus on panel data (and in particular the setup of synthetic control and synthetic interventions) due to its ubiquity and for concreteness.

\subsection{Our Contributions}
The first contribution of our work is a general framework for decision-making in the presence of strategic agents in the panel data setting. 
Our model is formally defined in Section~\ref{sec:setting}, but for the purposes of exposition, we outline it informally here as well. 
In our setting, a principal originally observes historical data from pre- and post- intervention outcomes of $n$ units. 
Each of these units $i \in \{1, \ldots, n\}$ is randomly assigned an intervention $d_i$ as in an A/B test (i.e., a randomized control trial), so they do not have incentive to strategize with their pre-intervention outcomes. 
After observing this historical data, the principal commits to an intervention policy $\pi$, which uses a unit's pre-intervention outcomes to assign its intervention during the post-intervention period. 
Then $m$ new units arrive and strategically modify their pre-treatment outcomes; units are allowed to move in a ball of radius $\delta$ around their true pre-intervention outcomes, but they are not further constrained. 
We call this behavior \emph{best-responding} and formally define it in Definition~\ref{ass:mod}. 
Finally, the principal observes the altered pre-intervention outcomes, assigns interventions according to policy $\pi$, and collects the post-intervention rewards. 
The goal of the principal in this setting is to deploy a policy $\pi$ that is \emph{strategyproof}, i.e., assigns the \khedit{utility-maximizing} intervention to each unit, despite the fact that they may have strategically modified their pre-intervention outcomes. 
We call this utility-maximizing intervention the unit's \emph{type}.

Given that units know the principal's policy $\pi$ and they are allowed to best respond anywhere within $\delta$ of their true pre-intervention outcome, it may seem like {\bf Q1} has a negative answer. 
However, in Section~\ref{sec:characterizing}, we derive the \emph{necessary and sufficient conditions} for a strategyproof intervention policy to exist.
In order to obtain this full characterization, we translate our the principal's problem of assigning interventions to the dual space (i.e., the space of the units' actions), and derive properties that units of the same type must share.

Our next contribution is to specialize our characterization of strategyproof intervention policies to the setting where unit outcomes at each round are determined by a latent factor model. 
Under a latent factor model (a natural and popular assumption in panel data settings), unit outcomes under different interventions are linearly dependent on a unit latent factor (i.e, their type), as well as a latent vector that is time- and intervention-specific. 
In this specification, the rewards of the principal are linear in their (expected) pre-intervention outcomes, which implies the units are {\em linearly separable} based on their type.
We show that our necessary and sufficient condition for a strategyproof intervention policy to exist is always satisfied when there are two interventions (\Cref{cor:two}), but it is in general \emph{not} satisfied for three or more interventions (\Cref{ex:impossible}). 
The intuition for both results is that in order for an intervention policy to be strategyproof under the latent factor model, the principal has to shift the decision boundaries by some amount in order to account for the fact that units are able to strategize. 
When the number of interventions is more than two, it may be the case that there are units of different types whose pre-intervention outcomes are ``close enough'' to each other such that no movement of the boundary can prevent all units from fooling the principal. 
Importantly, in Section~\ref{sec:learning} we show that assigning three or more interventions in a panel data setting with strategic agents can be interpreted as an instance of \emph{strategic multiclass classification}, a natural generalization of the well-studied strategic (binary) classification problem to the multiclass setting. 
As such, our impossibility result for strategyproof intervention policies translates to an impossibility result for strategyproof classification with three or more classes. 
To the best of our knowledge, we are the first to both draw this connection, and discuss strategic multiclass classification altogether. 

Having characterized strategyproof intervention policies under a latent factor model, we shift our focus to {\bf Q2}. For the case of two interventions (i.e a single treatment and control), we provide an algorithm for learning a strategyproof intervention policy from historical data (\Cref{alg:sp-two-action}). 
The analysis of~\Cref{alg:sp-two-action} relies on two steps: First, we upper-bound the difference between the reward of the learned policy and that of the optimal policy by the estimation error on the rewards of the test units (\Cref{cor:two-learn}). 
Second, we leverage bounds for estimation from the ``error-in-variables'' regression literature to derive end-to-end finite sample guarantees (\Cref{thm:out-of-sample}). 
We show that under relatively minor algebraic assumptions, our intervention policy is asymptotically optimal in the limit of infinite data. 
Next, we provide analogous finite sample guarantees for an extension of~\Cref{alg:sp-two-action} to the setting with an arbitrary number of treatments (\Cref{alg:learning-sp})---under an additional assumption on the difference in rewards between the optimal and next-best intervention for each type of unit (\Cref{ass:margin-app}). 
Relaxing this assumption on the reward gap appears challenging; we provide evidence for the necessity of such a condition in~\Cref{ex:1D}.

Finally, we complement our theoretical results with experiments based on panel data from product sales at several stores over the course of 18 months. 
We find the data is well-approximated by a latent factor model, and that the intervention policy of~\Cref{alg:sp-two-action} outperforms a baseline policy which does not take strategic interactions into consideration---even when the algorithm's estimate of $\delta$ (the unit effort budget) is misspecified. 
\subsection{Related Work}\label{sec:related}
Our work is broadly related to two lines of research in algorithmic game theory and econometrics: algorithmic decision-making under incentives and learning from panel data.
%
%\subsubsection{Strategic responses to algorithmic decision making} 
\paragraph{Strategic responses to algorithmic decision making}
A growing line of work at the intersection of computer science and economics aims to model the effects of using algorithmic assessment tools in high-stakes decision-making settings (e.g., \cite{hardt2016strategic,dong2018strategic, chen2020learning, kleinberg2020classifiers, shavit2020causal, munro2020learning, ahmadi2021strategic,bechavod2021gaming,bechavod2022information,ghalme2021strategic,harris2021bayesian,harris2021stateful,harris2022strategic,jagadeesan2021alternative, levanon2021strategic}). 
\citet{hardt2016strategic} introduce the problem of \emph{strategic classification}, in which a ``jury'' (principal) deploys a classifier, and a ``contestant'' (agent), best-responds by strategically modifying their observable features. 
Subsequent work has studied online learning settings \cite{dong2018strategic,chen2020learning,ahmadi2021strategic}, repeated interactions \cite{harris2021stateful}, social learning settings \cite{bechavod2022information}, and settings in which the model being used to make decisions is (partially) unknown to the strategic agents \cite{ghalme2021strategic, harris2021bayesian, bechavod2022information}.
Perhaps the line of work most relevant to ours is that of \cite{shavit2020causal,munro2020learning, bechavod2021gaming,harris2022strategic}, which aims to identify causal relationships between observable features and outcomes in the presence of strategic responses to various linear models. 
In contrast, we study a panel data setting in which the principal must assign one of several interventions to strategic units based on longitudinal data which may not have any underlying linear structure.
In Appendix~\ref{sec:sc}, we discuss the connections between our panel data setting and that of \emph{multiclass} strategic classification.
In particular, intervening on strategic units which exhibit a latent factor model structure may be viewed as a particular instance of multiclass classification where agents strategically modify their observable features. 
We are the first to study such a multiclass strategic classification setting, to the best of our knowledge, and we find that 
new ideas are required to handle the multiclass nature of the decision-making problem at hand.\looseness-1
%While the measurement error can be handled using standard methods from the error-in-variables literature, new ideas are required to handle the multiclass nature of the decision-making problem at hand.
%\aacomment{Better structure previous few sentences. Points are not coming through crisply.}

%\subsubsection{Panel data methods in econometrics}
\paragraph{Panel data methods in econometrics}
As stated earlier, this is a setting where one gets repeated measurements of multiple heterogeneous units over time. 
Prominent frameworks for causal estimation using panel data include \emph{difference-in-differences} \cite{ashenfelter1984using, bertrand2004much, harmless_econometrics} and \emph{synthetic control methods} \cite{abadie2003economic, abadie2010synthetic, Hsiao12, imbens16, athey1, LiBell17, xu_2017, rsc, mrsc, Li18, ark, bai2020matrix, asc, Kwok20, chernozhukov2020practical, fernandezval2020lowrank, agarwal2020robustness, agarwal2020principal}.
In these frameworks, there is a notion of a ``pre-intervention'' period where all the units are under control (i.e, no intervention), after which a subset of units receive one of many possible interventions.
The goal of these works is to estimate what would have happened to a unit that undergoes an intervention (i.e., a ``treated'' unit) if it had remained under control (i.e., no intervention), in the potential presence of unobserved confounding.
That is, they estimate the counterfactual if a treated unit remains under control for all $T$ time-steps.
A critical aspect that enables the methods above is the structure between units and time under control. 
One elegant encoding of this structure is through a \emph{latent factor model} (also known as an interactive fixed effect model), \cite{chamberlain, liang_zeger, arellano, bai03, bai09, pesaran, moon_15, moon_weidner_2017}.
In such models, it is posited that there exist low-dimensional latent unit and time factors that capture unit- and time-specific heterogeneity, respectively, in the potential outcomes.
Since the goal in these works is to estimate outcomes under control, no structure is imposed on the potential outcomes under intervention.
In \cite{SI, agarwal2021causal}, the authors extend this latent factor model to incorporate latent factorization across interventions as well, which  allows for identification and estimation of counterfactual mean outcomes under intervention rather than just under control.
In essence, we extend these previous works to allow for the pre-intervention outcomes to be strategically manipulated by units to receive a more favorable intervention.
What we find noteworthy is that the latent factor model typically assumed in these settings leads to strategyproof estimators that have a simple closed form.\looseness-1

\section{Model: Strategic Interactions in Panel Data Settings}\label{sec:setting}
%
%We now describe our framework for strategyproof decision-making in panel data settings.
%
\paragraph{Notation}
%\subsubsection{Notation} 
%
Subscripts are used to index the unit and time-step, superscripts are reserved for interventions. 
We use $i$ to index units, $t$ time-steps, and $d$ interventions.
For $x \in \mathbb{N}$, we use the shorthand $\range{x} := \{1, 2, \ldots, x\}$ and $\range{x}_0 := \{0, 1, \ldots, x-1\}$. 
%
%$\mathring{\vy}$ denotes a measurement of a vector $\vy$, corrupted with zero-mean noise.
%
%$\vbeta[j]$ denotes the $j$-th component of vector $\vbeta$.
%
%Finally, we use the notation $a \wedge b := \min\{a,b\}$ and $a \vee b := \max\{a,b\}$.\looseness-1
%
Finally, all proofs which are not in the main body may be found in the Appendix.

%\subsubsection{Decision making in panel data settings} 
\paragraph{Decision making in panel data settings} 
Consider a setting in which the principal observes the outcomes of $m$ units for $T$ time-steps, where $\pv{y}{d}_{i,t} \in \mathbb{R}$ is the outcome of unit $i$ at time $t$ under intervention $d$.
We assume that unit outcomes are generated via a \emph{latent factor model}, a popular assumption in the panel data setting (e.g., references in~\Cref{sec:related}). 
Most work in panel data with latent factor models assumes there is measurement noise in the outcomes. 
While we consider settings with measurement noise in~\Cref{sec:learning}, our results in~\Cref{sec:characterizing} do not depend on this noise and so we present a simpler setup here without noise for ease of exposition.
\begin{assumption}[Latent Factor Model]\label{ass:lfm}
The outcome for unit $i$ at time $t$ under treatment $d \in \range{k}_0$ is
\begin{equation*}
    y_{i,t}^{(d)} = \langle \vu_t^{(d)}, \mathbf{v}_i \rangle,
\end{equation*}
where $\vu_t^{(d)} \in \mathbb{R}^s$ is a latent factor which depends only on time $t$ and intervention $d$, and $\mathbf{v}_i \in \mathbb{R}^s$ is a latent factor which only depends on unit $i$.
For simplicity, we assume $|\pv{y}{d}_{i,t}| \leq 1$.\looseness-1
\end{assumption}
Note that~\Cref{ass:lfm} does not require the principal to know $\vu_t^{(d)}$ or $\mathbf{v}_i$\khedit{, and that $y_{i,t}^{(d)}$ is observed for only one intervention $d \in \range{k}_0$}. 
We assume that the latent dimension $s$ is known to the principal for ease of analysis, although several principled heuristics exist for estimating $s$ in practice from data (see, e.g. \cite[Section 2.2.1]{SI} for details).

Consider a pre-intervention period of $T_0$ time-steps, for which each unit is under the same intervention, i.e., under \emph{control}. 
After the pre-intervention period, the principal assigns an \emph{intervention} $d_i \in \range{k}_0$ to each unit $i \in \range{m}$. 
Without loss of generality, we denote control by $d=0$.
Once assigned intervention $d_i$, unit $i$ remains under $d_i$ for the remaining $T-T_0$ time-steps.  
We use 
\begin{equation*}
    \vy_{i,pre} := [y_{i,1}^{(0)}, \ldots, y_{i,T_0}^{(0)}]^\top \in \mathbb{R}^{T_0}
\end{equation*}
to refer to the set of unit $i$'s pre-treatment {\em observed} outcomes under control, and 
\begin{equation*}
    \pv{\vy}{d}_{i,post} := [y_{i,T_0+1}^{(d)}, \ldots, y_{i,T}^{(d)}]^\top \in \mathbb{R}^{T-T_0}
\end{equation*}
to refer to the set of unit $i$'s post-intervention {\em potential} outcomes under intervention $d$.
We denote the set of possible pre-treatment outcomes by $\cY_{pre}$. 
%
%An intervention policy $\pi: \mathcal{Y}_{pre} \rightarrow \range{k}_0$ is a mapping from pre-treatment outcomes to interventions.
%, and the set of possible \emph{measured} pre-treatment outcomes by $\mathring{\cY}_{pre}$. \aacomment{Make distinction between $\cY_{pre}$ and $\mathring{\cY}_{pre}$ clearer.}
%
%The mechanism by which the principal assigns interventions to units is known as their \emph{intervention policy}.\looseness-1

\begin{definition}[Intervention Policy]
An intervention policy $\pi: \mathcal{Y}_{pre} \rightarrow \range{k}_0$ is a (deterministic) mapping from pre-treatment outcomes to interventions.
\end{definition}

\iffalse
\begin{definition}[Equivalence region]
    An equivalence region for an intervention $d \in \range{k}_0$ under intervention policy $\pi$ is defined as the set of all $\vy_{pre} \in \mathcal{Y}_{pre}$ such that $\pi(\vy_{pre}) = d$.
\end{definition}
\fi 

For a given unit $i$, we denote the intervention assigned to them by intervention policy $\pi$ as $d_i^{\pi}$.\footnote{We sometimes use the shorthand $d_i = d_i^{\pi}$ when the intervention policy is clear from the context.}
Given an intervention policy $\pi$, units may have an incentive to strategically modify their pre-treatment outcomes in order to receive a more desirable intervention.
In our e-commerce example, this would correspond to users strategically modifying their engagement levels for the pre-intervention period (e.g., by artificially reducing their time spent on the platform), to ``trick'' the online marketplace into assigning them a higher discount than the one which would maximize the marketplace's revenue in the post-intervention period.
More generally, we study a game between a principal and a population of units. The principal moves first by committing to an intervention policy. 
Each unit then \emph{best-responds} to the given intervention policy by strategically modifying their pre-intervention outcomes as follows:
\begin{definition}[Strategic Responses to Intervention Policies]\label{ass:mod}
Assume that interventions are ordered in increasing unit preference (i.e., units prefer $d$ to $d'$ for $d > d'$).
Given an intervention policy $\pi: \mathcal{Y}_{pre} \rightarrow \range{k}_0$, unit $i$ best-responds to $\pi$ by modifying their pre-treatment outcomes as
%
% \begin{equation*}
% %
% \begin{aligned}
% %
% \Tilde{\vy}_{i,pre} \in \; &\arg\max_{\widehat{\vy}_{i,pre} \in \mathcal{Y}_{pre}} \;\; \pi(\widehat{\vy}_{i,pre})\\
% %
% \text{s.t.} \;\; &\| \widehat{\vy}_{i,pre} - \vy_{i,pre}\|_2 \leq \delta,
% %
% \end{aligned}
% %
% \end{equation*}
\begin{equation*}
\begin{aligned}
\Tilde{\vy}_{i,pre} \in \; &\arg\max_{\widehat{\vy}_{i,pre} \in \mathcal{Y}_{pre}} \;\; \pi(\widehat{\vy}_{i,pre})\\  
\text{s.t.} \;\; &\| \widehat{\vy}_{i,pre} - \vy_{i,pre}\|_2 \leq \delta,
\end{aligned}
\end{equation*}
where $\delta \in \mathbb{R}_{>0}$ is the unit effort budget and is known to the principal. 
We assume that if a unit is indifferent between two modifications, they chose the one which requires the smallest effort investment.
%
%\footnote{For simplicity, we focus on the setting where the effort budget is defined with respect to the $\ell_2$ norm, although analogous results may be obtained for other settings where the effort budget is defined with respect to any $p$-norm.\looseness-1}
%
\end{definition}
By~\Cref{ass:mod}, the goal of each unit is to obtain the most desirable intervention possible when interventions are assigned according to $\pi$, subject to the constraint that their modification is bounded in $\ell_2$ norm by $\delta$. 
Such budget assumptions are common in the literature on algorithmic decision making in the presence of strategic agents (e.g., \cite{chen2020learning,kleinberg2020classifiers,harris2021stateful, harris2023strategic}), and are useful for modeling ``hard constraints'' in a unit's ability to manipulate.
For example, in some settings the manipulation of pre-treatment outcomes may have some associated monetary cost, and units may have a fixed budget which they cannot exceed. 
In other settings the manipulation of pre-treatment outcomes may take time, and the $\delta$-ball represents the set of all possible pre-treatment outcomes a unit could achieve in the amount of time in the pre-treatment period.
Given~\Cref{ass:mod}, the goal of the principal is to design an intervention policy to maximize their \emph{reward} in the presence of such strategic manipulations.\looseness-1
\begin{definition}(Principal Reward)\label{def:reward}
The principal's reward for unit $i$ under intervention $d$ is a weighted sum of unit $i$'s outcomes in the post-treatment time period. 
Specifically
\begin{equation*}
    r_i^{(d)} = \sum_{t = T_0 + 1}^T \omega_t \cdot y_{i,t}^{(d)},
\end{equation*}
where $\omega_t \in \mathbb{R}$ for $t > T_0$ are known to the principal.
\end{definition}
Linear rewards can capture many settings; e.g. in e-commerce, the online marketplace may wish to maximize the total amount of user engagement on the platform in the post-intervention period (this corresponds to $\omega_{t} = 1$ for $t > T_0$).
\begin{definition}[Unit Type]\label{def:type}
    We say that unit $i$ is of \emph{type} $d$ (with respect to the principal's reward) if assigning them intervention $d$ maximizes the principal's reward. 
    Formally, unit $i$ is of type $d$ if 
    \begin{equation*}
        d \in \arg\max_{d' \in \range{k}_0} \pv{r}{d'}_i,
    \end{equation*}
    where $\pv{r}{d'}_i$ is defined as in~\Cref{def:reward}.
\end{definition}
While in general the principal's reward for a unit $i$ may be a function of \emph{all} of unit $i$'s outcomes (not just those in the post-intervention period), we only consider intervention policies which intervene after a fixed pre-treatment time period (for which all units are under control), in line with the synthetic interventions and synthetic controls literature.
%
%Note that 
As we show, the principal's reward for a given unit may be rewritten as a function of that unit's pre-treatment outcomes when an additional linear span assumption is satisfied.
\begin{assumption}[Linear Span Inclusion]\label{ass:span}
    Consider the latent factor model of~\Cref{ass:lfm}. 
    We assume that
    \begin{equation*}
        \sum_{t=T_0+1}^T \omega_t \cdot \pv{\vu}{d}_t \in \mathrm{span}\{\pv{\vu}{0}_1, \ldots, \pv{\vu}{0}_{T_0} \}.
    \end{equation*}
\end{assumption}
\Cref{ass:span} can be viewed as a form of ``causal transportability'' \emph{over time} which allows the principal to learn something about potential outcomes in the post-intervention time period from the pre-intervention time period. 
Such assumptions are fairly common in the literature on learning from panel data (e.g. \cite{amjad2018robust,SI}).

Next we show that under~\Cref{ass:lfm} and~\Cref{ass:span}, the principal's reward can be written as a function of their pre-intervention outcomes under control. 
This will be a useful structural condition for later results. 
\begin{lemma}[Reward Reformulation]\label{prop:reward}
Under~\Cref{ass:span}, $r_i^{(d)}$ may be rewritten as
\begin{equation*}
    r_i^{(d)} = \langle \pv{\vbeta}{d}, \vy_{i,pre} \rangle,
\end{equation*}
for some $\pv{\vbeta}{d} \in \mathbb{R}^{T_0}$, where $\vy_{i,pre}$ are the (unmodified) pre-intervention outcomes for unit $i$.
\end{lemma}
\begin{proof}
    From~\Cref{ass:lfm} and~\Cref{def:reward},
    \begin{equation*}
    \begin{aligned}
        r_i^{(d)} = \left \langle \sum_{t = T_0 + 1}^T \omega_t \cdot \vu_t^{(d)}, \mathbf{v}_i \right \rangle.
    \end{aligned}
    \end{equation*}
    Applying $\sum_{t=T_0+1}^T \omega_t \cdot \pv{\vu}{d}_t \in \mathrm{span}(\{\pv{\vu}{0}_1, \ldots, \pv{\vu}{0}_{T_0} \})$, we may rewrite $r_i^{(d)}$ as 
    \begin{equation*}
    \begin{aligned}
        r_i^{(d)} = \left \langle \sum_{t = 1}^{T_0} \pv{\beta}{d}_t \cdot \vu_t^{(0)}, \mathbf{v}_i \right \rangle
    \end{aligned}
    \end{equation*}
    for some $\pv{\vbeta}{d} = [\pv{\beta}{d}_1, \ldots, \pv{\beta}{d}_{T_0}]^\top \in \mathbb{R}^{T_0}$.
\end{proof}
Observe that any strategic modification by a unit in the pre-intervention period does \emph{not} change their latent factor $\mathbf{v}$ (or therefore, their post-intervention outcomes).
Given knowledge of a unit's latent factor, it would be trivial for the principal to assign them their \khedit{utility-maximizing} intervention. 
However, this knowledge is usually not available; instead it must be estimated from the unit's (strategically modified) pre-treatment behavior.
Therefore, we are interested in characterizing and learning intervention policies which assign the \khedit{utility-maximizing} intervention to each unit in the presence of strategic manipulations. 
Borrowing language from the game theory literature, we refer to such intervention policies as \emph{strategyproof}.\looseness-1
\begin{definition}[Strategyproof Intervention Policy]\label{def:spip}
An intervention policy $\pi$ is strategyproof if 
\begin{equation*}
    \pi(\Tilde{\vy}_{i,pre}) = \arg \max_{d \in \range{k}_0} \pv{r}{d}_i
\end{equation*}
for every unit $i$, where $\Tilde{\vy}_{i,pre} \in \mathbb{R}^{T_0}$ are unit $i$'s strategically-modified pre-treatment outcomes according to~\Cref{ass:mod}.
\end{definition}
See~\Cref{fig:summary} for a summary of the setting we consider.
\begin{figure}[htbp]
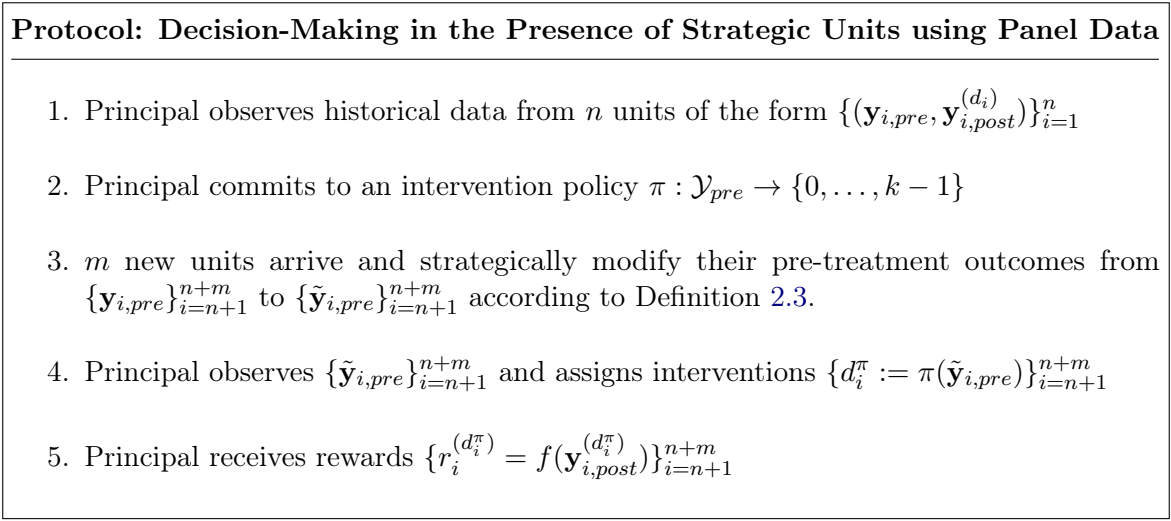

    \centering
    \noindent \fbox{\parbox{0.95\textwidth}{
    \vspace{0.1cm}
    \textbf{Protocol: Decision-Making in the Presence of Strategic Units using Panel Data}
    
    \vspace{-0.2cm}
    %\noindent\rule{15.9cm}{0.4pt}
    \noindent\rule{15.1cm}{0.4pt}
    %\noindent\rule{13.2cm}{0.4pt}\vspace{1mm}
    \begin{enumerate}
        \item Principal observes historical data from $n$ units of the form $\{(\vy_{i,pre}, \vy_{i,post}^{(d_i)})\}_{i=1}^n$\vspace{2mm}
        \item Principal commits to an intervention policy $\pi: \mathcal{Y}_{pre} \rightarrow \{0, \ldots, k-1\}$\vspace{2mm}
        \item $m$ new units arrive and strategically modify their pre-treatment outcomes from $\{\vy_{i, pre}\}_{i=n+1}^{n+m}$ to $\{\Tilde{\vy}_{i, pre}\}_{i=n+1}^{n+m}$ according to~\Cref{ass:mod}.\vspace{2mm} 
        \item Principal observes $\{\Tilde{\vy}_{i, pre}\}_{i=n+1}^{n+m}$ and assigns interventions $\{d_{i}^{\pi} := \pi(\Tilde{\vy}_{i, pre})\}_{i=n+1}^{n+m}$\vspace{2mm}
        \item Principal receives rewards $\{\pv{r}{d_{i}^{\pi}}_{i} = f(\pv{\vy}{d_i^{\pi}}_{i,post})\}_{i=n+1}^{n+m}$\vspace{1mm}
    \end{enumerate}
    }}
    \caption{Protocol for decision-making in our panel data setting under strategic responses.}
    \label{fig:summary}
\end{figure}
\subsection{Background on Principal Component Regression (PCR)}\label{sec:PCR-appendix}
Several of our results in~\Cref{sec:learning} leverage \emph{principal component regression (PCR)}~\citep{jolliffe1982note} to obtain end-to-end finite sample guarantees. 
At a high level, PCR learns a linear relationship between covariates and outcomes by first ``de-noising'' the covariates via hard singular value thresholding, then learning a linear relationship between the de-noised covariates and the outcome of interest. 
Formally, let $$\pv{Y}{d}_{pre} := [\vy_{i,pre}^\top : i \in \pv{\cN}{d}] \in \mathbb{R}^{\pv{n}{d} \times T_0}$$ and %$$Y'_{pre} := [\vy_{n+i,pre}^\top : i \in \range{m}] \in \mathbb{R}^{m \times T_0}.$$ 
denote the singular value decomposition of $\pv{Y}{d}_{pre}$ as $$\pv{Y}{d}_{pre} = \sum_{l=1}^{\pv{n}{d} \wedge T_0} \pv{s}{d}_l \pv{\widehat{\vu}}{d}_l (\pv{\widehat{\mathbf{v}}}{d}_l)^\top,$$ where $\pv{s}{d}_l \in \mathbb{R}$ is the $l$-th singular value, $\pv{\widehat{\vu}}{d}_l \in \mathbb{R}^{\pv{n}{d}}$ is the $l$-th left singular vector, and $\pv{\widehat{\mathbf{v}}}{d}_l \in \mathbb{R}^{T_0}$ is the $l$-th right singular vector. 
Denote the vector of observed principal rewards under intervention $d$ as $$\pv{\vr}{d} := [\pv{r}{d}_i: i \in \pv{\cN}{d}]^\top \in \mathbb{R}^{\pv{n}{d}}.$$ 
For a given hyperparameter $p \leq \pv{n}{d} \wedge T_0$, we can use PCR to estimate $\pv{\vbeta}{d}$ as 
\begin{equation*}
    \pv{\widehat{\vbeta}}{d} := \left( \sum_{l=1}^p \frac{1}{\pv{s}{d}_l} \pv{\widehat{\mathbf{v}}}{d}_l (\pv{\widehat{\vu}}{d}_l)^\top \right) \pv{\vr}{d}.\footnote{Our end-to-end learning results rely on the results of~\citet{agarwal2023adaptive}, which provide guarantees for \emph{regularized} principal component regression. However, this difference is not important for our purposes. We point the interested reader to~\citet[Section 3]{agarwal2023adaptive} for further details on regularized PCR.}
\end{equation*}
\section{Characterizing Strategyproof Intervention Policies}\label{sec:characterizing}
We now turn our attention to the problem of \emph{characterizing} strategyproof intervention policies. 
In other words, we are interested in (1) deriving conditions under which strategyproof intervention policies do or do not exist, and (2) characterizing the form of strategyproof intervention policies whenever they do exist. 
In~\Cref{sec:learning}, we focus on the problem of \emph{learning} strategyproof intervention policies from historical data.

We begin by deriving a necessary and sufficient condition for a strategyproof intervention policy to exist. 
In order to characterize this condition, we need to first introduce the notion of a \emph{best-response ball}. 
\begin{definition}[Best-Response Ball]\label{def:ball}
The best-response ball of a set of units $\cU$ is the set of all pre-intervention outcomes $\Tilde{\cY}_{pre}(\cU)$ such that %$\Tilde{\vy}_{pre} \in \Tilde{\cY}_{pre}(\cU)$ if $\left\| \Tilde{\vy}_{pre} - \vy_{i, pre} \right\|_2 \leq \delta$ for any unit $i \in \cU$, 
\begin{equation*}
\Tilde{\vy}_{pre} \in \Tilde{\cY}_{pre}(\cU) \; \text{ if } \; \left\| \Tilde{\vy}_{pre} - \vy_{i, pre} \right\|_2 \leq \delta \text{ for any unit } i \in \cU,
\end{equation*}
where $\vy_{i,pre} \in \cY_{pre}$ denotes the unmodified pre-intervention outcomes associated with unit $i$.
\end{definition}
The best-response ball for an \emph{individual unit} is its set of feasible modifications according to~\Cref{ass:mod}.
The best-response ball for a \emph{set of units} is the union of the best-response balls of all units contained within the set. 
%
%Note that (since $\cB(\cV)$ lies in the pre-intervention outcome space) in order for the principal to compute $\cB(\cV)$, they just need to know the unmodified pre-treatment outcomes associated with $\cV$. 
%
Equipped with this definition, we are now ready to introduce our sufficient and necessary condition, which we call \emph{separation of types}.\looseness-1

\begin{condition}[Separation of Types]\label{cond:nec-suf}
For a given problem instance, let $\pv{\cU}{d}$ denote the set of all units of type $d$ 
(recall~\Cref{def:type}).
Separation of types is satisfied if\looseness-1
\begin{equation*}
\forall d \in \range{k}_0, \; \not \exists \; i \in \pv{\cU}{d} \text{ s.t. } \Tilde{\cY}_{pre}(i) \subseteq \bigcup_{d' = 0}^{d-1} \Tilde{\cY}_{pre}(\pv{\cU}{d'}).
\end{equation*}
%
\iffalse
\begin{equation*}
%
\forall a \in \mathcal{A}, \; \not \exists \; i \in \pv{\cU}{a} \text{ s.t. } \cB(i) \subseteq \bigcup_{a' = 0}^{a-1} \cB(\pv{\cU}{a'}).
%
\end{equation*}
\fi 
%
\end{condition}
In other words, separation of types is satisfied if for all interventions $d \in \range{k}_0$, there does not exist any unit $i$ of type $d$ whose best-response ball $\Tilde{\cY}_{pre}(i)$ is a complete subset of the union of best-response balls of units with types less than $d$. 
\begin{theorem}\label{thm:impossible}
Separation of types (\Cref{cond:nec-suf}) is both necessary and sufficient for a strategyproof intervention policy (as defined in~\Cref{def:spip}) to exist.
\end{theorem}
\khedit{As the effort budget $\delta$ grows, both $\Tilde{\cY}_{pre}(i)$ and $\bigcup_{d' = 0}^{d-1} \Tilde{\cY}_{pre}(\pv{\cU}{d'})$ become larger, which makes~\Cref{cond:nec-suf} harder to satisfy.}
Necessity follows from leveraging~\Cref{def:ball} to show that if~\Cref{cond:nec-suf} does \emph{not} hold, there will always be at least one unit who can strategize to receive a better intervention.
We show sufficiency by giving a strategyproof intervention policy when~\Cref{cond:nec-suf} holds. 
The following two lemmas cover the necessity and sufficiency cases and immediately imply~\Cref{thm:impossible}. 
\begin{lemma}[Necessity]\label{lem:nec-app}
Suppose separation of types (\Cref{cond:nec-suf}) does not hold. 
Then there exists no mapping $\pi: \mathcal{Y}_{pre} \rightarrow \range{k}_0$ which can intervene perfectly on all unit types.
\end{lemma}
Intuitively, separation of types is \emph{necessary} because if it does not hold, then there always exists at least one unit of a lower type that can always ``pretend'' to be of higher type, thus leading the principal to intervene incorrectly on some subset of the population.
\begin{proof}
Assume that separation of types (\Cref{cond:nec-suf}) is violated for some unit $i$ of type $d$. 
Since $\cY_{pre}(i) \subseteq \bigcup_{d' = 0}^{d-1} \cY_{pre}(\pv{\cU}{d'})$, any valid modified pre-treatment behavior of unit $i$ can also be obtained by some other unit $i' \in \bigcup_{d' = 0}^{d-1} \pv{\cU}{d'}$ by~\Cref{def:ball}.
Therefore, no policy which assigns interventions according to a unit's observed pre-treatment outcomes can perfectly intervene on both $\pv{\cU}{d}$ and $\{\pv{\cU}{d'}\}_{d'=1}^{d-1}$.
\end{proof}

Next we show that separation of types is sufficient for a strategyproof intervention policy to exist, by providing a strategyproof intervention policy whenever separation of types holds. 
Recall that strategyproofness is defined with respect to whether the intervention assigned to a unit matches its type and \emph{not} with respect to whether modification of the pre-treatment outcomes takes place. 

\begin{lemma}[Sufficiency]\label{lem:suf-app}
Suppose separation of types (\Cref{cond:nec-suf}) holds. Then the following intervention policy is strategyproof:\looseness-1

Assign intervention $d_i$ to unit $i$, where 
\begin{equation}\label{eq:opt-ip-app}
d_i = \min\{d \in \range{k}_0 \; : \; \Tilde{\vy}_{i,pre} \in \cY_{pre}(\pv{\cU}{d}) \}
\end{equation}
\end{lemma}
\begin{proof}
No unit of type $d' < d$ can receive intervention $d$ by construction, since their pre-treatment outcomes will be in $\cY_{pre}(\pv{\cU}{d'})$ by definition. 
Therefore, it suffices to show that any unit of type $d$ can receive intervention $d$.

Consider a unit $i$ of type $d$. Since~\Cref{cond:nec-suf} holds, we know that there exists a vector of pre-treatment outcomes $\Tilde{\vy}_{pre} \in \cY_{pre}(i)$ such that
$\Tilde{\vy}_{pre} \not \in \bigcup_{d' = 0}^{d-1} \cY_{pre}(\pv{\cU}{d'})$. \looseness-1
Since $\cY_{pre}(i) \subseteq \cY_{pre}(\pv{\cU}{d})$, unit $i$ can receive intervention $d$ by strategically modifying their pre-treatment outcomes to $\Tilde{\vy}_{pre}$. 
\end{proof}
We will revisit the computational complexity of evaluating this policy later in the section. 
%
%The significance of~\Cref{thm:impossible} may not be obvious \emph{a priori}, as it is not immediately clear if/when~\Cref{cond:nec-suf} holds in our panel data setting of interest.
%
In the remainder of the section, we aim to shed some light on the significance of~\Cref{thm:impossible} by describing situations in which it does/does not hold in our panel data setting of interest. 
We begin by showing that~\Cref{cond:nec-suf} always holds in the important special case when there is only a single treatment and control.
\begin{theorem}\label{cor:two}
If $d \in \{0, 1\}$, separation of types (\Cref{cond:nec-suf}) always holds under~\Cref{ass:lfm} and~\Cref{ass:span}. 
Moreover, the following closed-form intervention policy is strategyproof:

Assign intervention $d_i$ to unit $i$, where
\begin{equation}\label{eq:two}
\begin{aligned}
    d_i &= 
    \begin{cases}
        1 \;\; \text{ if } \;\; \langle \pv{\vbeta}{1} - \pv{\vbeta}{0}, \Tilde{\vy}_{i,pre} \rangle - \delta \|
        \pv{\vbeta}{1} - \pv{\vbeta}{0} \|_2 > 0\\
        0 \;\; \text{ otherwise}
    \end{cases}
\end{aligned}
\end{equation}
We call $\langle \pv{\vbeta}{1} - \pv{\vbeta}{0}, \Tilde{\vy}_{i,pre} \rangle - \delta \| \pv{\vbeta}{1} - \pv{\vbeta}{0} \|_2 = 0$ the \emph{decision boundary} between interventions $0$ and $1$.\looseness-1
\end{theorem}
\begin{proof}
If $d \in \{0, 1\}$,we can simplify~\Cref{cond:nec-suf} to 
\begin{equation*}
\not \exists \; i \in \pv{\cU}{1} \text{ s.t. } \cY_{pre}(i) \subseteq \cY_{pre}(\pv{\cU}{0}).
\end{equation*}
By the reward reformulation (\Cref{prop:reward}), $\pv{\cU}{0}$ and $\pv{\cU}{1}$ are separated by a single hyperplane: $\langle \pv{\vbeta}{1} - \pv{\vbeta}{0}, \Tilde{\vy}_{pre} \rangle = 0$. 
Therefore, by the definition of best-response ball (\Cref{def:ball}), this simplified version of separation of types must always hold, and a strategyproof intervention policy may be obtained by shifting the hyperplane $\langle \pv{\vbeta}{1} - \pv{\vbeta}{0}, \Tilde{\vy}_{pre} \rangle = 0$ by $\delta$ (the unit effort budget) in the direction of $\pv{\vbeta}{0} - \pv{\vbeta}{1}$. Note that such an intervention policy is strategyproof since there exists at least one valid modification of pre-treatment outcomes for all units of type $1$ to receive treatment (namely, $\Tilde{\vy}_{pre} = \vy_{pre} + \delta \cdot (\pv{\vbeta}{1} - \pv{\vbeta}{0}) / \|\pv{\vbeta}{1} - \pv{\vbeta}{0}\|_2$), and there exists no valid modification of pre-treatment outcomes for any unit of type $0$ to receive treatment (due to~\Cref{ass:mod}).
\end{proof}

\iffalse
\begin{figure}[t]
%
\centering
%
\includegraphics[width=\textwidth]{figure/combined_fig.png}
%
\caption{
%
{\bf Left:} The optimal policy in the non-strategic setting assigns the control (intervention $0$) to the units left of the blue line and intervention $1$ to the units to the right. 
%
{\bf Center (\Cref{cor:two})}: When units are strategic, the optimal decision boundary is shifted by $\delta$ in the direction of the decision boundary.
%
{\bf Right:} A setting with three interventions for which no strategyproof intervention policy exists.\looseness-1}
%
\label{fig:two-intervention-visual}
%
\end{figure}
\fi 
%
\begin{figure}[t]
\centering
\includegraphics[width=\textwidth]{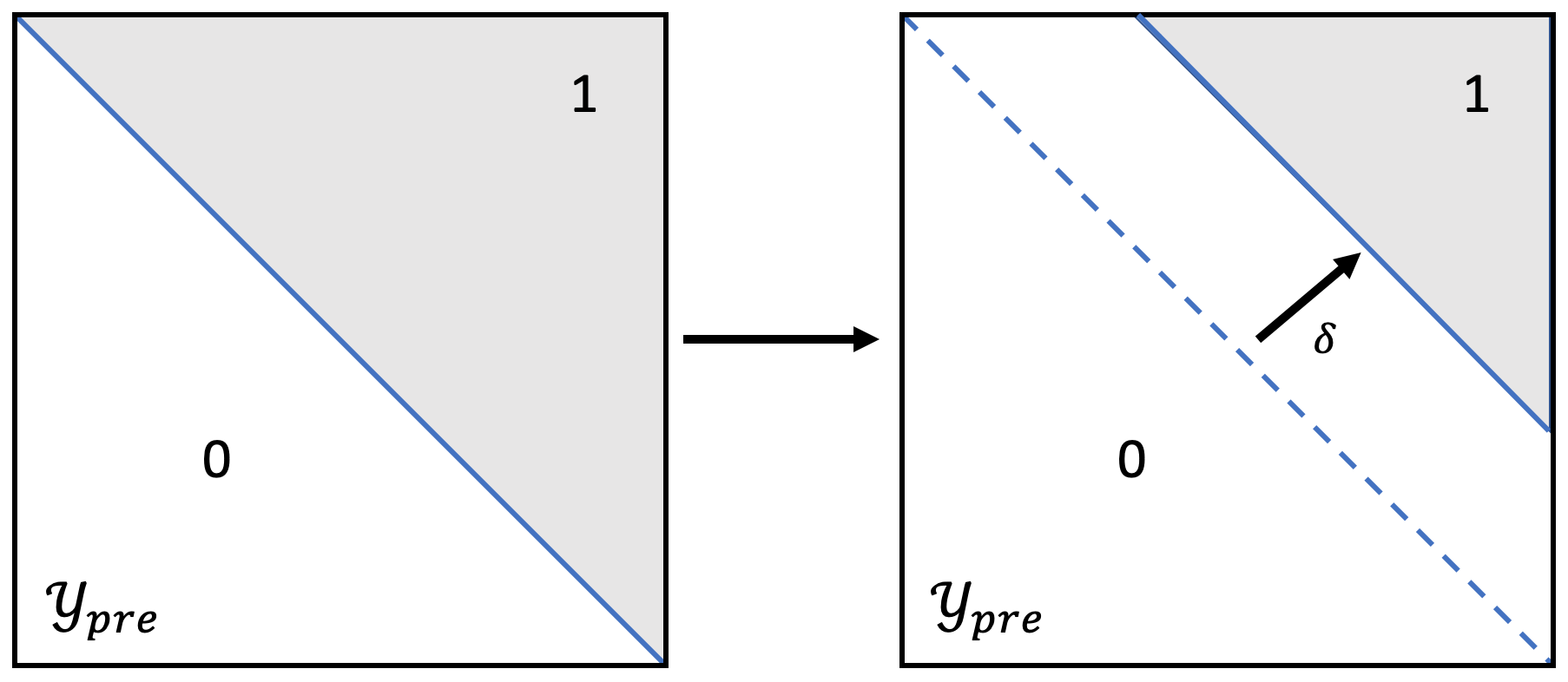}
\caption{Visualization of the space of pre-intervention outcomes for a space with two interventions. 
In this example, we assume that units prefer intervention $1$ to intervention $0$. 
Left (\Cref{prop:reward}): The optimal policy in the non-strategic setting assigns the control (intervention $0$) to the units left of the line $\langle \pv{\vbeta}{1} - \pv{\vbeta}{0}, \vy_{pre} \rangle = 0$ (in blue) and intervention $1$ to the units to the right. 
Right (\Cref{cor:two}): When units are strategic, the optimal decision boundary is shifted by $\delta$ in the direction of $\pv{\vbeta}{0} - \pv{\vbeta}{1}$.\looseness-1}
\label{fig:two-intervention-visual}
\end{figure}
See~\Cref{fig:two-intervention-visual} for an illustration of a strategyproof intervention policy in the binary intervention setting. 
Intuitively, the idea of~\Cref{cor:two} is to shift the true decision boundary in such a way to account for potential manipulations.
While this shift prevents units from ``gaming'' the policy to receive the intervention when it is not in the principal's best interest, it may require some units who should receive the intervention to strategize in order to do so.
Perhaps somewhat surprisingly, this line of reasoning does not carry over to the setting where there are more than two treatments.
\begin{theorem}\label{ex:impossible}
    There exists an instance with three interventions such that \Cref{cond:nec-suf} is not satisfied. 
\end{theorem}
\begin{figure}[t]
\centering
\includegraphics[width=0.47\textwidth]{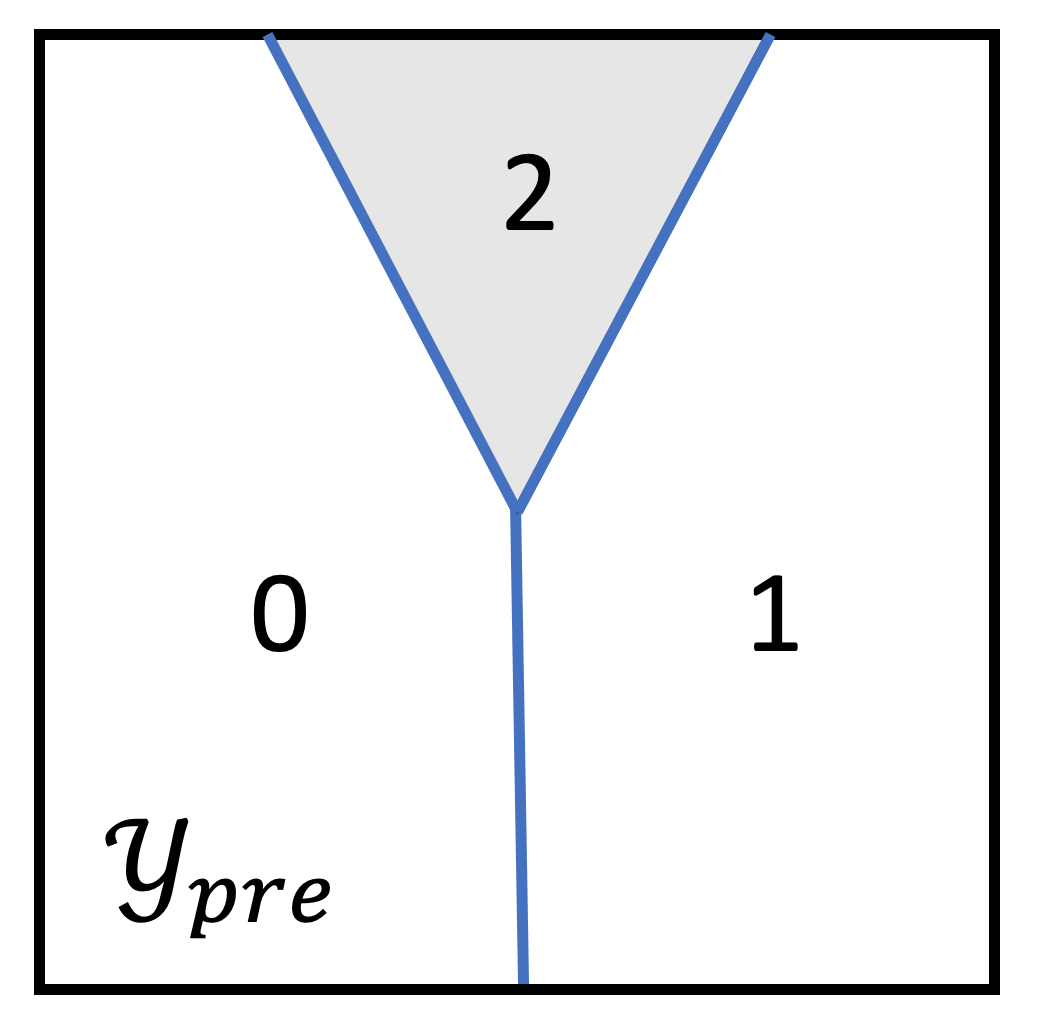}
\caption{Visualization of the space of pre-intervention outcomes for a setting with three interventions for which separation of types (\Cref{cond:nec-suf}) does \emph{not} hold.  
In this example, we assume that units prefer intervention $2$ to interventions $1$ and $0$. 
This implies that correctly intervening on \emph{all} (strategic) units is not possible under this setting.}
\label{fig:bad-example}
\end{figure}
See~\Cref{fig:bad-example} for a visualization of one such setting. 
At a high level,~\Cref{cond:nec-suf} cannot hold in this example since the decision boundaries $\langle \pv{\vbeta}{2} - \pv{\vbeta}{0}, \Tilde{\vy}_{i,pre} \rangle = 0$ and $\langle \pv{\vbeta}{1} - \pv{\vbeta}{0}, \Tilde{\vy}_{i,pre} \rangle = 0$ must both be shifted by $\delta$ in order to prevent some units from strategizing to receive intervention $2$.
However, this prevents other units who should receive intervention $2$ from receiving it, since the amount that they would need to modify their pre-intervention outcomes under these shifts is strictly greater than $\delta$.
We conclude this section by providing a strategyproof intervention policy for an arbitrary number of interventions whenever separation of types is satisfied.

\begin{theorem}\label{thm:three-plus}
    When~\Cref{cond:nec-suf} is satisfied, the following intervention policy is strategyproof and can be evaluated in time polynomial in $T_0$ and $k$ under~\Cref{ass:lfm} and~\Cref{ass:span}:

    Assign intervention $d_i$ to unit $i$, where %$d_i = \min\{d \in \range{k}_0 \; : \; \Tilde{\vy}_{i,pre} \in \Tilde{\cY}_{pre}(\pv{\cU}{d}) \}$.
    \begin{equation*}\label{eq:opt-ip}
    d_i = \min\{d \in \range{k}_0 \; : \; \Tilde{\vy}_{i,pre} \in \Tilde{\cY}_{pre}(\pv{\cU}{d}) \}.
    \end{equation*}
\end{theorem}
\begin{proof}[Proof sketch]
The form of the above intervention policy follows from the proof of sufficiency in~\Cref{thm:impossible}.
We show that membership to the set $\Tilde{\cY}_{pre}(\pv{\cU}{d})$ can be checked by solving a (convex) quadratic program (QP) with polynomial size, which implies that $d_i$ can be computed by solving at most $k$ such QPs.
\end{proof}
\begin{proof}
    Observe that by using the reward reformulation (\Cref{prop:reward}), the definition of a best-response ball (\Cref{def:ball}) may be rewritten as
    \begin{equation*}
    \begin{aligned}
        \Tilde{\vy}_{pre} \in \cY_{pre}(\cU) \; \text{ if } \; &\left\| \Tilde{\vy}_{pre} - \vy_{pre} \right\|_2 \leq \delta\\ \text{ for any } \vy_{pre} \in \cY_{pre} \text{ such that } \; \langle &\pv{\vbeta}{d} - \pv{\vbeta}{d'}, \vy_{pre} \rangle \geq 0 \; \text{ for all } \; d' \in \range{k}_0.
    \end{aligned}
    \end{equation*}
    Therefore, $\Tilde{\vy}_{pre} \in \cY_{pre}(\cU)$ if and only if $\pv{OPT}{d}$ is at most $\delta$, where 
    \begin{equation}\label{eq:optimization-app}
    \begin{aligned}
    \pv{OPT}{d} := \min_{\widehat \vy_{pre}} \;\; &\| \widehat \vy_{pre} - \Tilde{\vy}_{pre} \|_2 \\
    \text{s.t.} \;\; \langle &\pv{\vbeta}{d} - \pv{\vbeta}{d'}, \widehat \vy_{pre} \rangle \geq 0 \; \text{ for all } \; d' \in \range{k}_0.
    \end{aligned}
    \end{equation}
\end{proof}
\subsection{Implications for Strategic Multiclass Classification}\label{sec:sc}
We now highlight an impossibility result for the problem of strategic multiclass classification, which readily follows from~\Cref{ex:impossible} and may be of independent interest. 
For the reader who is uninterested in strategic classification, this subsection may be skipped without any loss in continuity. 

\paragraph{Background on strategic classification}
%\subsubsection{Background on strategic classification} 
%
When subjugated to algorithmic decision making, decision subjects (agents) have an incentive to strategically modify their input to the algorithm in order to receive a more desirable prediction. In the context of machine learning models, such settings have been formalized in the literature under the name of \emph{strategic classification} (see, e.g., \cite{hardt2016strategic, dong2018strategic, chen2020learning}). In the strategic (binary) classification setting, the principal commits to an \emph{assessment rule} (usually a linear model), which maps from \emph{observable features} to \emph{binary predictions}. Using knowledge of the assessment rule, strategic agents may modify their observable features in order to maximize their chances of receiving a desirable classification, subject to some constraint on the amount of modification which is possible (e.g., a best-response analogous to our~\Cref{ass:mod}). Given an agent's modified features, the principal uses their assessment rule to make a prediction about the agent. After the prediction is made, the principal receives some feedback about how accurate the assessment rule's prediction was. Under such a setting, the goal of the principal is to deploy an assessment rule with high accuracy on strategic agents.\looseness-1

%\subsubsection{Impossibility result} 
\paragraph{Impossibility result} 
Using ideas similar to those used in~\Cref{ex:impossible}, we show that an impossibility result holds for the \emph{multiclass} generalization of the strategic (binary) classification setting, where each strategic agent now belongs to one of $k \geq 3$ classes (as opposed to the binary setting, where $k=2$). 
We consider a setting in which a \emph{principal} interacts with $m$ strategic \emph{agents}. 
Each agent $i$ has a set of \emph{initial} observable features $\vy_{i} \in \mathcal{Y}$\footnote{Our notation is different than the standard one adopted in the strategic classification literature in order to best match the notation from the rest of the paper.}. 
These features are privately observable by the agents and they are not revealed to the principal. 
Instead, the agents report features $\Tilde{\vy}_{i} \in \mathcal{Y}$ to the principal, which may be strategically modified. 
Given observed features $\Tilde{\vy}_{i}$, the principal makes a \emph{prediction} $d_{i}$ from some set of possible \emph{classes} $\range{k}_0$. 
We assume that each agent has some true label $d_{i}^*$, and the principal receives reward $1$ if $d_{i} = d_{i}^*$ and reward $0$ otherwise.
In contrast to the principal's reward, we assume that each agent's reward $r^A_{i}(d)$ is a function of the prediction alone, i.e., $r^A_{i}(d) = r^A(d), \forall i \in \range{m}$, and is known to the principal.\footnote{This generalizes the typical assumption made in strategic binary classification, where it is assumed that agents prefer the positive label to the negative one.}

The principal's \emph{assessment rule} $\pi: \mathcal{Y} \rightarrow \range{k}_0$ is a mapping from observable features to predictions. 
In particular, given a set of training data consisting of $\{(\vy_i, d_i)\}_{i=1}^n$ pairs from $n$ non-strategic agents, the goal of the principal is to deploy an assessment rule which minimizes the \emph{out-of-sample} error on $m$ strategic units.

\begin{figure}
    \centering
    \noindent \fbox{\parbox{0.95\textwidth}{
    \vspace{0.1cm}
    \textbf{Protocol: strategic multiclass classification}
    
    \vspace{-0.2cm}
    \noindent\rule{15.1cm}{0.4pt}
    %\noindent\rule{13.2cm}{0.4pt}
    \begin{enumerate}
        \vspace{1mm}
        \item Using historical data collected from $n$ non-strategic agents, the principal learns and publicly commits to an assessment policy $\pi: \mathcal{Y} \rightarrow \range{k}_0$
        \vspace{2mm}
        \item $m$ new agents arrive and strategically modify their observable features from $\vy_{i}$ to $\Tilde{\vy}_{i}$ according to~\Cref{ass:br-sc}, for $i \in \{n+1, \ldots, n+m\}$
        \vspace{2mm}
        \item Principal observes $\Tilde{\vy}_{i}$ and assigns prediction $d_{i} = \pi(\Tilde{\vy}_{i})$ to agent $i$
        \vspace{2mm}
        \item Principal receives reward $\pv{r}{d_{i}}_{i} = \mathbbm{1}\{d_{i} = d_{i}^*\}$
        \vspace{1mm}
    \end{enumerate}
    }}
    \caption{Summary of the strategic multiclass classification setting we consider.}
    \label{fig:summary-sc}
\end{figure}

\begin{definition}[Out-of-sample error]
    The out-of-sample error of a policy $\pi$ is defined as the empirical probability that $\pi$ makes an incorrect prediction on the $m$ test agents. Formally,
    \begin{equation*}
        \frac{1}{m} \sum_{i=n+1}^{n+m} \mathbbm{1}\{d_{i} \neq d_{i}^*\}
    \end{equation*}
\end{definition}

Given a principal policy $\pi$, it is natural for an agent to modify their observable features in a way which maximizes their reward. 
Specifically, we assume that agent $i$ strategically modifies their observable features based on the principal's policy, subject to a constraint on the amount of modification which is possible. 
In addition to being a common assumption in the strategic classification literature, this \emph{budget} constraint on the amount an agent can modify their features reflects the fact that agents have inherent constraints on the amount of time and resources they can spend on modification. 
\begin{assumption}[Agent Best Response]\label{ass:br-sc}
    We assume that agent $i$ best-responds to the principal's policy $\pi$ in order to maximize their expected reward, subject to the constraint that their modified observable features $\Tilde{\vy}_{i}$
    are within an $\ell_2$ ball of radius $\delta$ of their initial observable features $\vy_{i}$. Formally, we assume that agent $i$ solves the following optimization to determine their modified observable features:
    \begin{equation*}
        \begin{aligned}
            \Tilde{\vy}_{i} \in &\arg \max_{\widehat{\vy}_{i} \in \mathcal{Y}} \; r^A(\pi(\widehat{\vy}_{i}))\\
            \text{s.t.} \; &\|\widehat{\vy}_{i} - \vy_{i} \|_2 \leq \delta
        \end{aligned}
    \end{equation*}
\end{assumption}

Furthermore, we assume that if an agent is indifferent between modifying their observable features and not modifying, they choose not to modify. 
This assumption is analogous to how we define the unit best response in~\Cref{def:ball}. 
See Figure \ref{fig:summary-sc} for a summary of the setting we consider. 
We are now ready to present our main result for strategic multiclass classification, which follows straightforwardly from~\Cref{ex:impossible} and is visualized in~\Cref{fig:bad-example}.

\begin{corollary}
Suppose $\mathcal{Y} = \mathbb{R}^2$, $k = 3$, and $r^A(2) > r^A(1) = r^A(0)$. For the following labeling over $\mathcal{Y}$, there exists a distribution over agents such that no policy can achieve perfect classification if agents strategically modify according to~\Cref{ass:br-sc}:
\begin{equation*}
    d_{i}^* = 
    \begin{cases}
        0 \;\; &\text{if } \langle \vbeta_{10}, \vy_{i} \rangle < 0 \text{ and } \langle \vbeta_{20}, \vy_{i} \rangle < 0\\
        1 \;\; &\text{if } \langle \vbeta_{21}, \vy_{i} \rangle < 0 \text{ and } \langle \vbeta_{10}, \vy_{i} \rangle \geq 0\\
        2 \;\; &\text{if } \langle \vbeta_{21}, \vy_{i} \rangle \geq 0 \text{ and } \langle \vbeta_{20}, \vy_{i} \rangle \geq 0\\
    \end{cases}
\end{equation*}
where $\vbeta_{20} = [1 \; 0.5]^\top$, $\vbeta_{21} = [-1 \; 0.5]^\top$, and $\vbeta_{10} = [2 \; 0]^\top$.
\end{corollary}
\section{Learning Strategyproof Intervention Policies}\label{sec:learning}

\begin{algorithm}[t]
    \SetAlgoNoLine
    \SetAlgoNoEnd
    \vspace{1mm}
     \KwIn{Trajectories $\{(\vy_{i,pre}, \vy_{i,post}^{(0)})\}_{i \in \pv{\cN}{0}}$, $\{(\vy_{i,pre}, \vy_{i,post}^{(1)})\}_{i \in \pv{\cN}{1}}$} 
    \vspace{2mm}
    Compute $r_i^{(d_i)} = \sum_{t = T_0 + 1}^T \omega_t \cdot y_{i,t}^{(d_i)}$ for $i \in \range{n}$.\\
    \vspace{2mm}
    For $d \in \{0, 1\}$, use $\{(\vy_{i,pre}, \pv{r}{d}_i)\}_{i \in \pv{\cN}{d}}$ to estimate $\pv{\vbeta}{d}$ as $\pv{\widehat{\vbeta}}{d}$.\\
    \vspace{2mm}
    \For{$i = n+1, \ldots, n+m$}
    {
        \vspace{2mm}
        Assign intervention
        \begin{equation*}
        %\tag*{(SPI)}
            d_{i}^A = 
            \begin{cases}
                1 \; &\text{if } \langle \pv{\widehat{\vbeta}}{1} - \pv{\widehat{\vbeta}}{0}, \Tilde{\vy}_{i,pre} \rangle - \delta \|\pv{\widehat{\vbeta}}{1} - \pv{\widehat{\vbeta}}{0} \|_2 > 0\\
                0 \; &\text{otherwise.}
            \end{cases}
        \end{equation*}
    }
    \caption{Learning Strategyproof Interventions with One Treatment}
    \label{alg:sp-two-action}
\end{algorithm}
We now shift our focus from characterizing strategyproof intervention policies to \emph{learning} them from historical data.
We assume that the historical data has \emph{not} been strategically modified, as is the case when, e.g., interventions are assigned according to a \emph{randomized control trial} (since in such settings, units do not have an incentive to strategize). 
While~\Cref{thm:three-plus} provides a characterization of a strategyproof intervention policy when one exists, deploying such an intervention policy requires knowledge of the underlying relationships between pre-treatment outcomes and principal rewards, which may not be known \emph{a priori}. 
Additionally, it may be unreasonable to assume that the latent factor model holds \emph{exactly}, due to measurement error or randomness in the outcomes of each unit.\footnote{In our e-commerce example, noise may be due to randomness in day-to-day engagement with the platform.\looseness-1} With this in mind, we overload the notation of $\pv{y}{d}_{i,t}$ and consider the following relaxation of~\Cref{ass:lfm} throughout the sequel.
\begin{assumption}[Latent Factor Model; revisited]\label{ass:lfm-r}
Suppose the outcome for unit $i$ at time $t$ under treatment $d \in \range{k}_0$ takes the following factorized form:
\begin{equation*}\label{eq:lfm-r}
\begin{aligned}
\E[y_{i,t}^{(d)}] &= \langle \vu_t^{(d)}, \mathbf{v}_i \rangle\\ 
\text{and } \;\; \pv{y}{d}_{i,t} &= \E[y_{i,t}^{(d)}] + \varepsilon_{i,t},
\end{aligned}
\end{equation*}
where $\vu_t^{(d)} \in \mathbb{R}^s$ and $\mathbf{v}_i \in \mathbb{R}^s$ are defined as in~\Cref{ass:lfm}, and $\varepsilon_{i,t}$ is zero-mean sub-Gaussian random noise with variance at most $\sigma^2$. 
For simplicity, we assume that $|\E[y_{i,t}^{(d)}]| \leq 1$.
\end{assumption}
Note that under~\Cref{ass:lfm-r}, the reward reformulation (\Cref{prop:reward}) now holds \emph{in expectation}.
Inspired by the linear form of the strategyproof intervention policy of~\Cref{cor:two} for two interventions, we begin by deriving performance guarantees for a ``plug-in'' version of this intervention policy. Our algorithm proceeds as follows: 
Given historical trajectories of the form $\{(\vy_{i,pre}, \vy_{i,post}^{(d)})\}_{i \in \pv{\cN}{d}}$ for each $d \in \{0,1\}$, we can calculate the principal reward for assigning intervention $d_i$ to unit $i$ as 
\begin{equation*}
    r_i^{(d_i)} := \sum_{t = T_0 + 1}^T 
\omega_t \cdot y_{i,t}^{(d_i)} \; \text{ for }\; i \in \range{n},
\end{equation*}
where $\pv{\cN}{d}$ denotes the set of historical (non-strategic) units who received intervention $d$. %and $\pv{n}{d}$ denotes the size of $\pv{\cN}{d}$ (i.e., $\pv{n}{d} := |\pv{\cN}{d}|$). 
Given the $(\vy_{i,pre}, \pv{r}{d}_i)$ pairs as training data,~\Cref{alg:sp-two-action} uses an error-in-variables regression method (e.g., principal component regression \cite{pcr_jolliffe, pcr_tibshirani}) to estimate $\pv{\vbeta}{0}, \pv{\vbeta}{1}$. 
The last step is to use the estimated linear coefficients to construct a ``plug-in'' estimator of intervention policy (\ref{eq:two}) to use when assigning interventions to the $m$ (strategic) out-of-sample units.
\begin{theorem}\label{cor:two-learn}
    Suppose $d \in \{0, 1\}$, $d_i^A$ is the intervention assigned to unit $i$ by~\Cref{alg:sp-two-action}, $d_i^*$ is the optimal intervention to assign to unit $i$, and 
    \begin{equation*}
        \pv{\widehat{r}}{d}_{i} := \langle \pv{\widehat{\vbeta}}{d}, \vy_{i,pre} \rangle
    \end{equation*}
    is the estimated principal reward under intervention $d$.
    Then for any test unit $n+i$, 
    %
    \iffalse
    \begin{equation}\label{eq:two-learn}
        \frac{1}{m} \sum_{i \in \range{m}} \left(\E [\pv{r}{d_{n+i}^A}_{n+i}] - \E [\pv{r}{d_{n+i}^*}_{n+i} ] \right)^2 \leq \frac{4}{m} \max_{d \in \{0, 1\}} \sum_{i \in \range{m}} \left(\pv{\widehat{r}}{d}_{n+i} - \E [\pv{r}{d}_{n+i}] \right)^2
    \end{equation}
    \fi 
    %
    \begin{equation}\label{eq:two-learn}
        \E [\pv{r}{d_{n+i}^A}_{n+i}] - \E [\pv{r}{d_{n+i}^*}_{n+i} ] \leq |\pv{\widehat{r}}{0}_{n+i} - \E [\pv{r}{0}_{n+i}]| + |\pv{\widehat{r}}{1}_{n+i} - \E [\pv{r}{1}_{n+i}]|
    \end{equation}
\end{theorem}

\Cref{cor:two-learn} shows that the difference in performance of~\Cref{alg:sp-two-action} and a strategyproof intervention policy that assigns interventions optimally can be bounded by the difference between the actual and estimated rewards under each intervention. 
Therefore, if $\pv{\widehat{\vbeta}}{0}, \pv{\widehat{\vbeta}}{1}$ are good estimates $\pv{\vbeta}{0}, \pv{\vbeta}{1}$,~\Cref{alg:sp-two-action} will perform well. 
In~\Cref{sec:apply-pcr}, we leverage principal component regression to estimate $\vbeta^{(0)}$ and $\vbeta^{(1)}$ and obtain high-probability finite sample guarantees. 
Since we are dealing with strategically manipulated data, we are unable to apply prior results for learning from panel data in a black-box way.
Our key insight which enables us to obtain performance guarantees for Algorithm~\ref{alg:sp-two-action} is that its performance is matched by another intervention policy which makes decisions on units \emph{which are not strategic} (intervention policy \ref{eq:not-sp-app} in the Appendix). 
Given this observation, the bound follows readily from algebraic manipulation.\looseness-1

Next we show that analogous performance guarantees can be obtained for the extension of~\Cref{alg:sp-two-action} to the setting where there are more than two interventions, when there is a sufficiently large \emph{gap in the principal's expected rewards} for each unit type. 
This property is natural in many settings of interest; in our e-commerce running example, it corresponds to the principal deriving very different rewards from offering a discount that is not optimal for each group. 
We now present performance guarantees for~\Cref{alg:learning-sp}, which is an extension of~\Cref{alg:sp-two-action} to settings with more than two interventions.\looseness-1
\begin{corollary}[Informal; detailed version in Corollary~\ref{cor:multi-app}]\label{cor:multi} 
    For $\alpha > 0$, suppose the principal's expected rewards satisfy a sufficiently large reward gap $g(\alpha)$ (\Cref{ass:margin-app}). Then,\looseness-1
    \iffalse    
    \begin{equation*}
        \frac{1}{m} \sum_{i \in \range{m}} \left(\E[\pv{r}{d_{n+i}^A}_{n+i}] - \E[\pv{r}{d_{n+i}^*}_{n+i}] \right)^2 \leq \frac{k^2}{m} \max_{d \in \range{k}_0} \sum_{i \in \range{m}} \left(\pv{\widehat{r}}{d}_{n+i} - \E[\pv{r}{d}_{n+i}] \right)^2
    \end{equation*}
    \fi 
    \begin{equation*}
        \E[\pv{r}{d_{n+i}^A}_{n+i}] - \E[\pv{r}{d_{n+i}^*}_{n+i}] \leq \sum_{d=0}^{k-1} |\pv{\widehat{r}}{d}_{n+i} - \E[\pv{r}{d}_{n+i}]|
    \end{equation*}
    with probability at least $1 - \alpha$ for any test unit $n + i$, where $d_{n+i}^A$ is the intervention assigned to unit $n+i$ by~\Cref{alg:learning-sp}, and $d_{n+i}^*$ and $\pv{\widehat{r}}{d}_{n+i}$ are defined as in~\Cref{cor:two-learn}.\looseness-1
\end{corollary}
\begin{algorithm}[t]
        \SetAlgoNoLine
        \SetAlgoNoEnd
        \vspace{1mm}
        \KwIn{Trajectories $\{\{(\vy_{i,pre}, \vy_{i,post}^{(d)})\}_{i \in \pv{\cN}{d}}\}_{d=0}^{k-1}$}
        \vspace{2mm}
        Compute $r_i^{(d_i)} = \sum_{t = T_0 + 1}^T \omega_t \cdot y_{i,t}^{(d_i)}$ for $i \in \range{n}$.\\
        \vspace{2mm}
        For $d \in \range{k}$, use $\{(\vy_{i,pre}, \pv{r}{d}_i)\}_{i \in \pv{\cN}{d}}$ to estimate $\pv{\vbeta}{d}$ as $\pv{\widehat{\vbeta}}{d}$.\\
        \vspace{2mm}
        \For{$i = n+1, \ldots, n+m$}
        {
            \vspace{2mm}
            Assign intervention $d_i^B = d$ to unit $i$ if 
            \begin{equation*}
            \begin{aligned}
                \langle \pv{\widehat{\vbeta}}{d} - \pv{\widehat{\vbeta}}{d'}, \Tilde{\vy}_{i,pre} \rangle - \delta \| \pv{\widehat{\vbeta}}{d} - \pv{\widehat{\vbeta}}{d'} \|_2 &> 0 \; \text{ for all } \; d' < d\\
                \text{and } \; \langle \pv{\widehat{\vbeta}}{d} - \pv{\widehat{\vbeta}}{d'}, \Tilde{\vy}_{i,pre} \rangle + \delta \| \pv{\widehat{\vbeta}}{d} - \pv{\widehat{\vbeta}}{d'} \|_2 &\geq 0 \; \text{ for all } \; d' > d\\
            \end{aligned}
            \end{equation*}
        }
        \caption{Learning Strategyproof Interventions under the Expected Reward Gap Assumption}
        \label{alg:learning-sp}
\end{algorithm}

Intuitively, a gap assumption is not needed in the single treatment regime since a unit will only modify their pre-treatment behavior in order to receive the (single) treatment.
This is in contrast to the multi-treatment setting, where a unit's best response may be in one of several directions depending on which treatment(s) they are capable of receiving under a particular intervention policy. 
Obtaining performance guarantees for learning algorithms which do not require a gap assumption appears challenging for the general case, as the unit best response is not guaranteed to converge smoothly as $\{\pv{\widehat{\vbeta}}{d}\}_{d=0}^{k-1}$ approaches $\{\pv{\vbeta}{d}\}_{d=0}^{k-1}$.
%
%Nevertheless, it would be interesting to explore whether such a gap assumption is necessary, or if it could be weakened or removed completely using a different algorithm or analysis.
%
\subsection{Evidence Towards the Necessity of a Gap}
To build intuition as to why such a gap may be necessary, consider the following example.\looseness-1
\begin{figure}
    \centering
    \includegraphics[width=0.6\textwidth]{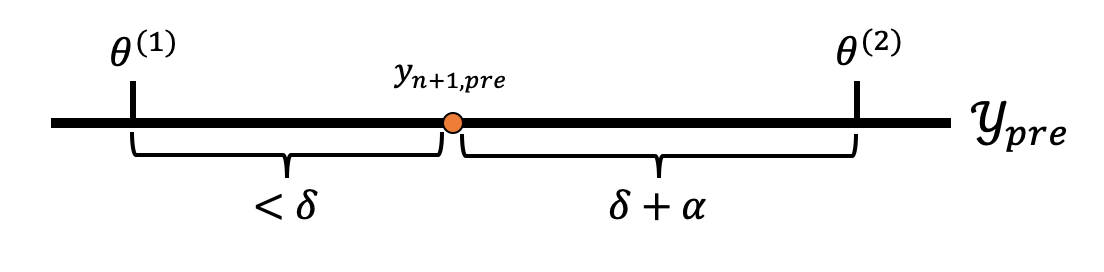}
    \caption{Visualization of the setting in~\Cref{ex:1D}. If $\pv{\theta}{1}$ and $\pv{\theta}{2}$ are perfectly known to the principal, unit $n+1$ will modify their pre-intervention outcome to receive intervention $1$ (left of $\pv{\theta}{1}$). However if the principal's estimate of $\pv{\theta}{2}$ is sufficiently inaccurate, unit $n+1$ may be able to modify their pre-intervention outcome to receive intervention $2$ (right of $\pv{\theta}{2}$).}
    \label{fig:delta-example}
\end{figure}
\begin{example}\label{ex:1D}
    Consider the one-dimensional setting in~\Cref{fig:delta-example}, where the unit can manipulate by $\delta$ in either direction. Specifically, let $\mathcal{Y}_{pre} \in \mathbb{R}$, $d \in \{0, 1, 2\}$, and let the optimal intervention policy\footnote{We do not specify the principal's reward function so as to simplify the exposition.} be:\looseness-1
    \begin{equation}\label{eq:opt-ex}
        d_i^* = 
        \begin{cases}
            1 \; \text{ if } \; \Tilde{y}_{i,pre} \leq \pv{\theta}{1}\\
            2 \; \text{ if } \; \Tilde{y}_{i,pre} \geq \pv{\theta}{2}\\
            0 \; \text { o.w.}
        \end{cases}
    \end{equation}
    for some $\pv{\theta}{1} \in \mathbb{R}$, $\pv{\theta}{2} \in \mathbb{R}$, and $\pv{\theta}{1} < \pv{\theta}{2}$.  Suppose the principal deploys the following ``plug-in'' estimate of intervention policy (\ref{eq:opt-ex}):
    \begin{equation*}
        d_i^P = 
        \begin{cases}
            1 \; \text{ if } \; \Tilde{y}_{i,pre} \leq \pv{\widehat{\theta}}{1}\\
            2 \; \text{ if } \; \Tilde{y}_{i,pre} \geq \pv{\widehat{\theta}}{2}\\
            0 \; \text { o.w.}
        \end{cases}
    \end{equation*}
    for $\pv{\widehat{\theta}}{1} \in \mathbb{R}$ and $\pv{\widehat{\theta}}{2} \in \mathbb{R}$. Moreover, suppose that the principal has perfect knowledge of $\pv{\theta}{1}$ (i.e., $\pv{\widehat{\theta}}{1} = \pv{\theta}{1}$) and after observing data from $n$ non-strategic units, $|\pv{\widehat{\theta}}{2}(n) - \pv{\theta}{2}| = c/n$ for some $c > 0$.
    Now consider a strategic unit $n+1$ with $y_{n+1,pre} = y$ such that $y - \pv{\theta}{1} < \delta$ and $\pv{\theta}{2} - y = \delta + \alpha$, for some $\alpha > 0$. If the principal had perfect knowledge of $\pv{\theta}{2}$, unit $n+1$'s best-response would be to modify their pre-treatment outcome to $\pv{\theta}{1}$ and receive intervention $1$. However, given the expression for $\pv{\widehat{\theta}}{2}$, we can write the best-response of unit $n+1$ as
    \begin{equation*}
        \Tilde{y}_{n+1,pre} = 
        \begin{cases}
            \pv{\theta}{1} \; &\text{ if } \; \pv{\widehat{\theta}}{2} = \pv{\theta}{2} + \frac{c}{n}\\
            \pv{\widehat{\theta}}{2} (= y + \delta + \alpha - \frac{c}{n}) \; &\text{ if } \; \pv{\widehat{\theta}}{2} = \pv{\theta}{2} - \frac{c}{n}\\
        \end{cases}
    \end{equation*}
    as long as the number of non-strategic units $n < \frac{c}{\alpha}$. Therefore under this setting, if the $(n+1)$-st unit has pre-intervention outcome $y$, then their best-response may be highly discontinuous as long as $n < \frac{c}{\alpha}$ (recall that $\alpha$ can be chosen to be arbitrarily small), despite the fact that $\pv{\theta}{1}$ is perfectly known to the principal and $\pv{\widehat{\theta}}{2}$ converges to $\pv{\theta}{2}$ at the ``fast'' rate of $\mathcal{O}(1/n)$.  
\end{example}
\subsection{Application of PCR to Obtain End-to-End Guarantees}\label{sec:apply-pcr}
In order to leverage the out-of-sample guarantees for PCR, we make the following assumptions on the latent factors. %$\{\E[\pv{Y}{d}_{pre}]\}_{d=0}^{k-1}$ and $\E[Y'_{pre}]$ (as defined in~\Cref{sec:PCR-appendix}).\looseness-1
\begin{assumption}[Well-balancing assumptions]
\label{ass:spectrum}
The following hold: (a) $\sigma_i(\E[Y_{pre}^{(d)}]) = \Theta\left(\sqrt{\frac{c_n(d)T_0}{r}}\right)$ for all $i \in [r]$, (b) $c_n(d) = \Theta(c_n(d'))$ for all $d, d' \in \range{k}_0$, and (c) $k = O(s)$.

\end{assumption}
\Cref{ass:spectrum} says that (a) the non-zero singular values of the expected pre-intervention outcomes for the set of units who received each intervention are roughly equal, (b) each intervention is assigned roughly the same number of times, and (c) the number of interventions is on the order of the dimension of the latent subspace. 
(Part (c) is trivially satisfied whenever there is a single treatment and control.)
We are now ready to state our formal result for the convergence rates of~\Cref{alg:sp-two-action} using PCR. An analogous bound may be obtained for~\Cref{alg:learning-sp}.

\iffalse
\begin{corollary}\label{thm:out-of-sample}
Consider $\pv{\vbeta}{d} \in \emph{rowspan}(\E[\pv{Y}{d}_{pre}])$ for $d \in \range{k}_0$ and the procedure of~\Cref{alg:learning-sp}, where $\{\pv{\widehat{\vbeta}}{d}\}_{d=0}^{k-1}$ are given by PCR with $p = s$. 
%
Let $\|\pv{\vbeta}{d}\|_2 = \Omega(1)$ and $\|\pv{\vbeta}{d}\|_1 = \mathcal{O}(\sqrt{T_0})$ for all $d \in \range{k}_0$. If Assumptions \ref{ass:lfm-r} (latent factor model), \ref{ass:margin-app} (expected reward gap), \ref{ass:sub-inc} (subspace inclusion), \ref{ass:in} (balanced spectra) hold, then with probability at least $1 - \mathcal{O}(\alpha + \sum_{d=0}^{k-1}((\pv{n}{d} \wedge m) T_0)^{-10})$,\looseness-1
\begin{equation*}
\begin{aligned}
    &\frac{1}{m} \sum_{i=n+1}^{n+m} \left(\E[\pv{r}{d_{i}^A}_{i}] - \E[\pv{r}{d_{i}^*}_{i}] \right)^2 \leq k^2 \max_{d \in \range{k}_0} C_{noise}s^3 \log((\pv{n}{d} \wedge m) T_0)\\ &\cdot \left( \left( \frac{1 \vee \frac{T_0}{m}}{\pv{n}{d} \wedge T_0} + \frac{\pv{n}{d} \vee T_0}{(\pv{n}{d} \wedge T_0)^2} + \frac{1}{m} \right)\|\pv{\vbeta}{d}\|_1^2 + \left( \frac{\sqrt{\pv{n}{d}}}{\pv{n}{d} \wedge T_0} \right)\|\pv{\vbeta}{d}\|_1 \right)
\end{aligned}
\end{equation*}
\end{corollary}
\fi 
\begin{corollary}\label{thm:out-of-sample}
Let $\delta \in (0, 1)$ be an arbitrary confidence parameter and regularization parameter $\rho > 0$ be chosen to be sufficiently small, as detailed in~\citet{agarwal2023adaptive}. 
Further, assume that Assumptions \ref{ass:lfm-r} (latent factor model), \ref{ass:margin-app} (expected reward gap), and \ref{ass:span} (subspace inclusion) are satisfied, and rank$(\mathbb{E}[Y_{pre}^{(d)}]) = s$ for all $d \in \range{k}_0$. 
If $T_0 \leq \frac{1}{2} T$, then under~\ref{ass:spectrum} (balanced spectra) with probability at least $1 - \mathcal{O}(\delta)$,
\begin{equation*}
\begin{aligned}
    &\E[\pv{r}{d_{n+i}^A}_{n+i}] - \E[\pv{r}{d_{n+i}^*}_{n+i}] \leq \Tilde{\mathcal{O}} \left( \frac{r^2}{T_0 \wedge n} + \frac{r}{\sqrt{T_0 \wedge n}} + \frac{\sqrt{T_0}}{\sqrt{T-T_0}}\right)
\end{aligned}
\end{equation*}
for any test unit $n+i$.
\end{corollary}
\noindent \Cref{thm:out-of-sample} follows immediately from applying~\citet[Theorem 5.5]{agarwal2023adaptive} to~\Cref{cor:two-learn}.
\section{Experiments}\label{sec:experiments}
We empirically evaluate the performance of~\Cref{alg:sp-two-action} on panel data constructed using time-series measurements of product sales at several stores. Our goal is to evaluate the performance of our methods when (1) data is \emph{not} generated using a latent factor model and (2) the principal has \emph{imperfect knowledge} about the units' ability to modify their pre-intervention outcomes.

%\subsubsection{Setup}
\paragraph{Setup}
Our initial dataset consists of weekly sales data from three products at nine different stores over the course of 18 months.\footnote{The dataset we use can be found at \href{https://raw.githubusercontent.com/susanli2016/Machine-Learning-with-Python/master/data/Sales_Product_Price_by_Store.csv}{\texttt{https://raw.githubusercontent.com/susanli2016/}}\\
\href{https://raw.githubusercontent.com/susanli2016/Machine-Learning-with-Python/master/data/Sales_Product_Price_by_Store.csv}{\texttt{Machine-Learning-with-Python/master/data/Sales\_Product\_Price\_by\_Store.csv}}.}
We consider two interventions: \texttt{discount} (the product is on sale) and \texttt{no discount} (the product is not on sale). 
We define a unit to be a (store, product) pair which was under \texttt{no discount} for five consecutive weeks, followed by either \texttt{discount} or \texttt{no discount} for three consecutive weeks. 
\khedit{Since we only have access to historical data, we only get to see post-intervention outcomes under only one of \texttt{discount} or \texttt{no discount} for any given unit. 
Therefore we use these (unit, intervention, outcome) tuples to} run a synthetic interventions procedure to generate counterfactual outcomes for all units under both \texttt{discount} and \texttt{no discount}.  
We use the resulting trajectories as the ground-truth rewards for each unit under both interventions.

In order to train our model, we randomly assign interventions to $50\%$ of the units (135 trajectories), and we use the remaining $50\%$ to test the performance.
Under such a setting, strategic behavior may arise when, for example, a local store manager wishes to maximize the number of products sold at their specific location, while the owner of the store chain ultimately wants to maximize revenue. 
In this case, the local store manager could conceivably have an incentive to strategically misreport their weekly revenue during the pre-treatment time period so that their products are given a discount and their sales increase (as was the case for Zara, previously mentioned in~\Cref{sec:intro}).

\begin{table}
    \caption{Average normalized change in revenue and standard deviation over 10 runs for various estimates of $\delta$ (denoted by $\widehat{\delta}$). Note that $\widehat{\delta} = 0$ corresponds to the naive policy which does not take any strategic interactions into consideration and $\widehat{\delta} = \delta$ corresponds to the intervention policy of~\Cref{alg:sp-two-action}.}
	\label{tab:results}
	\begin{minipage}{\columnwidth}
		\begin{center}
			\begin{tabular}{lll}
				$\widehat{\delta} / \delta$ & Normalized $\Delta$ Revenue & Standard Deviation\\
                    \toprule
                    0 (Naive Policy) & 0.237 & 0.110\\
                    %0.1 & 0.404 & 0.087\\
                    0.2 & 0.527 & 0.126\\
                    0.5 & 0.831 & 0.033\\
                    1 (\Cref{alg:sp-two-action}) & 0.989 & 0.011\\
                    2 & 0.943 & 0.014\\
                    5 & 0.846 & 0.024\\
                    %10 & 0.845 & 0.032\\
				\bottomrule
			\end{tabular}
		\end{center}
	\end{minipage}
\end{table}

%\subsubsection{Results} 
\paragraph{Results} 
See~\Cref{tab:results} for a summary of our results. 
For an intervention policy $\pi$, we are interested in the increase in revenue from assigning interventions according to $\pi$, as opposed to the alternative. We normalize with respect to the \emph{optimal} improvement in revenue, i.e. the best possible improvement if the principal were able to observe both counterfactual trajectories before assigning an intervention. Denote the intervention assigned by policy $\pi$ to unit $n+i$ as $d_{n+i}^{\pi}$ and the intervention not assigned by $\pi$ to unit $n+i$ as $\neg d_{n+i}^{\pi}$. Formally,
\begin{equation*}
    \text{Normalized } \Delta \text{ Revenue} := \frac{\sum_{i \in \range{m}} \left( \pv{r}{d_{n+i}^{\pi}}_{n+i} - \pv{r}{\neg d_{n+i}^{\pi}}_{n+i}\right)}{\sum_{i \in \range{m}} \left( \pv{r}{d_{n+i}^{*}}_{n+i} - \pv{r}{\neg d_{n+i}^{*}}_{n+i} \right)} %\left. \sum_{i \in \range{m}} \left( \pv{r}{d_{n+i}^{\pi}}_{n+i} - \pv{r}{\neg d_{n+i}^{\pi}}_{n+i}\right) \middle/  \sum_{i \in \range{m}} \left( \pv{r}{d_{n+i}^{*}}_{n+i} - \pv{r}{\neg d_{n+i}^{*}}_{n+i} \right) \right. .
\end{equation*}
Note that Normalized $\Delta$ Revenue is at most $1$.
Since the unit effort budget $\delta$ may be unknown in practice, we also examine the performance of~\Cref{alg:sp-two-action} when the principal's estimate of $\delta$ (the unit's effort budget; defined in~\Cref{def:ball}) is misspecified as $\widehat \delta$. 
We find that the intervention policy of~\Cref{alg:sp-two-action} is able to achieve near-optimal improvement in revenue, in contrast to the relatively poor performance of the naive policy which does not consider incentives. 
Additionally, we observe that, under the experimental setup we consider, the performance of~\Cref{alg:sp-two-action} degrades gracefully as a function of model misspecification (as quantified by $\widehat \delta / \delta$). 
\khedit{Our empirical results suggest that if $\delta$ is unknown to the principal, it may be better for them to use an overestimate instead of an underestimate.} 
\section{Conclusions and Future Work}
%
\iffalse
We provide both algorithms and impossibility results for strategyproof decision-making in various panel data settings when units are strategic.
%
There are several exciting directions for future work.\looseness-1

\subsubsection{Heterogeneous unit preferences.} Separation of types (\Cref{cond:nec-suf}) relies on the form of the unit's best response. 
%
Our model of homogeneous unit preferences captures a variety of settings (e.g. customers generally prefer higher discounts to lower ones), but it would be interesting to study more general settings which allow for heterogeneous unit preferences over interventions.

\subsubsection{Truthful mechanisms} Our strategyproof mechanisms may require some units to strategize in order to be assigned the intervention that matches their true type. 
%
It would be interesting to design mechanisms for incentivizing \emph{truthful} behavior in panel data settings, i.e. incentivizing units to \emph{not} alter their data at all when reporting their pre-intervention outcomes, although doing so appears challenging.\looseness-1 
%
%While we provide such a truthful mechanism in~\Cref{sec:learning} under a gap assumption on principal rewards, 
%
%deriving such a truthful mechanism without any gap assumption appears challenging.\looseness-1
\fi
We introduce a framework for strategy-aware decision-making in panel data settings. In settings captured by our framework, we provide a sufficient and necessary condition for a strategyproof intervention policy to exist. 
Next we specialize our results to the canonical setting where unit outcomes are generated via a latent factor model, and the principal's reward is a linear combination of post-intervention outcomes. 
Under this setting, we show that the strategyproof intervention policy takes a simple closed form, when one exists. 
Additionally when there is only a single treatment and control, we show that a strategyproof intervention policy always exists, and we provide an algorithm for learning such an intervention policy from historical data. 
We also show that analogous performance guarantees can be obtained in the general setting under a gap assumption on the principal's rewards.
Finally, we provide concrete rates of convergence for learning a strategyproof intervention policy when the parameters of interest are estimated via principal component regression. 
Along the way, we prove impossibility results for strategic multiclass classification which may be of independent interest.
There are several exciting directions for future work.\looseness-1 
\paragraph{Gap-free learning with multiple treatments}
%
%\subsubsection{Gap-free learning with multiple treatments} 
%
In order to obtain convergence rates for the intervention policy of~\Cref{alg:learning-sp}, our analysis relies on a gap assumption between the rewards of different interventions. It would be interesting to further explore if such a gap is indeed necessary, or if a tighter analysis or different algorithm could be used to weaken or remove this assumption.\looseness-1
\paragraph{Heterogeneous unit preferences and effort.}
%
%\subsubsection{Heterogeneous unit preferences and effort.} %Our sufficient and necessary condition for a strategyproof intervention policy to exist relies on the form of the unit's best-response. 
%
Our model of homogeneous unit preferences and effort budgets captures a variety of settings (e.g., patients may prefer an effective-but-costly medical treatment to a less-effective-but-cheaper alternative, customers generally prefer higher discounts to lower ones). 
However, it would be interesting to study more general games between the principal and strategic units which allow for, e.g., heterogeneous effort budgets and unit preferences over interventions.

% While our model of homogeneous unit preferences and a common effort budget captures some settings well (e.g., patients may prefer an effective-but-costly medical treatment to a less-effective-but-cheaper alternative, customers generally prefer higher discounts to lower ones), it may not be realistic under others. It would be interesting to study more general games between the principal and strategic units which allow for, e.g., heterogeneous effort budgets and unit preferences over interventions.
%
%\subsubsection{Truthful mechanisms} 
\paragraph{Truthful mechanisms} 
Recall that our strategyproof mechanisms may require some units to strategize in order to be assigned the intervention that matches their true type. 
It would be interesting to explore the feasibility of designing mechanisms for incentivizing \emph{truthful} behavior in panel data settings, i.e., incentivizing the agents to not alter their data at all when reporting their pre-intervention outcomes. 
While we provide such a truthful mechanism in~\Cref{sec:learning} under a gap assumption on principal rewards, 
deriving a truthful mechanism without any gap assumption appears challenging, as intervening based on a unit's pre-intervention outcomes provides an incentive for units to behave non-truthfully.
\khedit{
%\subsubsection{Robustness to assumptions and unknown parameters}
\paragraph{Robustness to assumptions and unknown parameters}
While our empirical results in~\Cref{sec:experiments} suggest that our methods are fairly robust to data which is not generated by a latent factor model as well as overestimates of the unit effort budget $\delta$, a more thorough theoretical robustness analysis is an important step towards deploying our methods in real-world panel data settings. 
}
\section*{Acknowledgements}
KH is supported in part by an NDSEG Fellowship. ZSW is supported in part by the NSF FAI Award \#1939606. The authors would like to thank the anonymous reviewers for valuable feedback and Hoda Heidari for helpful comments and suggestions in early stages of the project.
\newpage

\bibliographystyle{ACM-Reference-Format}
\bibliography{refs}

\newpage
\appendix
\section{Decision-Making using Synthetic Interventions is not Strategyproof}\label{sec:not-sp}
We begin by showing that assigning interventions to strategic units 
using the \emph{synthetic interventions} method is generally not strategyproof, even when there is only a single treatment and control. The setting we consider is that of~\Cref{fig:summary}, with $d \in \{0, 1\}$, principal reward equal to the sum of post-intervention outcomes, and intervention policy $\pi$ given by the following synthetic interventions procedure:
\begin{enumerate}
    \item For $d \in \{0, 1\}$, learn a linear model $\pv{\widehat{\boldsymbol{\omega}}}{d}_{i} \in \mathbb{R}^{\pv{n}{d}}$ between unit $i$ and $\pv{\cN}{d}$ for $i \in \{n+1, \ldots, n+m\}$ using, e.g., PCR, and the pre-treatment outcomes of unit $i$ and $\pv{\cN}{d}$.
    \item For $i \in \{n+1, \ldots, n+m\}$, estimate the reward of assigning intervention $d$ to unit $i$ as 
    \begin{equation*}
        \pv{\widehat{r}}{d}_{i} = \sum_{t=T_0 + 1}^T \sum_{j \in \pv{\cN}{d}} \pv{\widehat{\boldsymbol{\omega}}}{d}_{i}[j] \pv{y}{d}_{j,t}
    \end{equation*}
    \item Assign interventions $\{d_{i} := \pi(\Tilde{\vy}_{i, pre}) = \arg \max_{d \in \range{k}_0} \pv{\widehat{r}}{d}_{i}\}_{i=n+1}^{n+m}$
\end{enumerate}

\begin{figure}
    \centering
    \includegraphics[width=0.25\textwidth]{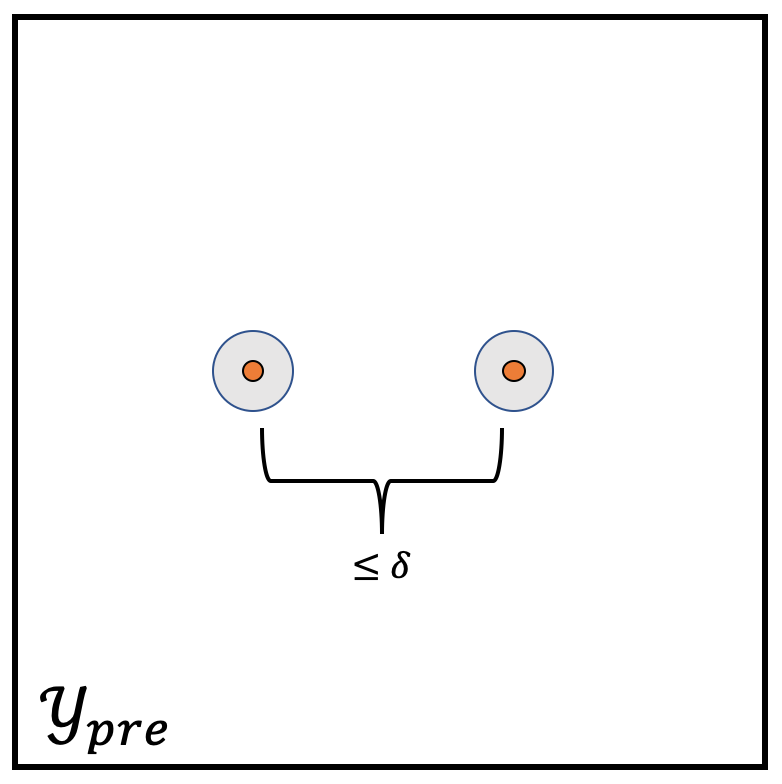}
    \caption{Illustration of~\Cref{ex:si}. The left and right orange dots are $\E[\pv{\vy}{0}_{pre}]$ and $\E[\pv{\vy}{1}_{pre}]$ respectively. Grey circles represent the variance in the observed pre-treatment outcomes due to measurement noise. If $\|\E[\pv{\vy}{1}_{pre}] - \E[\pv{\vy}{0}_{pre}]\|_2 \leq \delta$, a constant fraction of units with private type $\pv{\mathbf{v}}{0}$ will be able to modify their pre-treatment outcomes to look like units with private type $\pv{\mathbf{v}}{1}$.}
    \label{fig:si-example}
\end{figure}

\begin{example}\label{ex:si}
For simplicity, consider $d \in \{0, 1\}$. Suppose there are only two unit types: $\pv{\mathbf{v}}{0}$ and $\pv{\mathbf{v}}{1}$, where $\E[\pv{r}{0}_{i}] > \E[\pv{r}{1}_{i}]$ if $\mathbf{v}_{i} = \pv{\mathbf{v}}{0}$ and $\E[\pv{r}{1}_{i}] > \E[\pv{r}{0}_{i}]$ if $\mathbf{v}_{i} = \pv{\mathbf{v}}{1}$. Furthermore, suppose that all units prefer receiving the treatment to the control, the decision maker has access to data from a randomized control trial (RCT) with $\pv{\cN}{0} = \pv{\cN}{0} = n/2$, and that the distribution over units is such that $\mathbb{P}(\mathbf{v}_i = \pv{\mathbf{v}}{0}) = \mathbb{P}(\mathbf{v}_i = \pv{\mathbf{v}}{1}) = 1/2$, for $i \leq n$. Intuitively, all units with private type $\pv{\mathbf{v}}{0}$ will have expected pre-treatment outcomes $\E[\pv{\vy}{0}_{pre}]$ and units with private type $\pv{\mathbf{v}}{1}$ will have expected pre-treatment outcomes $\E[\pv{\vy}{1}_{pre}]$. See~\Cref{fig:si-example} for a visual depiction of such a setting. Now consider a decision maker who uses the RCT data to construct a synthetic interventions procedure, and uses it to assign interventions to $m$ new units. If $n$ is sufficiently large and the new units are not strategic, then under standard assumptions the above synthetic interventions procedure should assign the optimal interventions to new units with constant probability, as each sequence of pre-treatment outcomes $\vy_{i,pre}$ is a noisy observation of either $\E[\pv{\vy}{0}_{pre}]$ or $\E[\pv{\vy}{1}_{pre}]$. However, if units are strategic and $\|\E[\pv{\vy}{1}_{pre}] - \E[\pv{\vy}{0}_{pre}]\|_2 \leq \delta$, then new units with private type $\pv{\mathbf{v}}{0}$ will be able to modify their pre-treatment behavior according to~\Cref{ass:mod} to look like units with private type $\pv{\mathbf{v}}{1}$, which will cause the synthetic interventions procedure to assign intervention $d=1$ to those units with constant probability. Therefore with constant probability, 
\begin{equation*}
    \frac{1}{m} \sum_{i=n+1}^{n+m} \left(\E[\pv{r}{d_{i}}_{i}] - \E[\pv{r}{d_{i}^*}_{i}] \right)^2 = \Omega(1)
\end{equation*}
when interventions are assigned using synthetic interventions, even as $n, T, m \rightarrow \infty$.
\end{example}
\section{Proofs from Section~\ref{sec:characterizing}}
\begin{theorem}\label{ex:impossible-app}
    There exists an instance with three interventions such that \Cref{cond:nec-suf} is not satisfied. 
\end{theorem}
\begin{proof}
Suppose $d \in \{0, 1, 2\}$, $T_0 = 2$, and units prefer intervention $2$ over interventions $1$ and $0$, of which they are indifferent between. 
Suppose that
\begin{equation*}
\pv{\vbeta}{0} = [-1 \;\; 0.5]^\top, \;\;\; \pv{\vbeta}{1} = [1 \;\; 0.5]^\top, \;\;\; \pv{\vbeta}{2} = [0 \;\; 1]^\top, \text{ and } \;\; \vy_{i,pre} = \mathbf{v}_i.
\end{equation*}

\noindent Consider the following set of unit types: Let
\begin{equation*}
    \pv{\cU}{0} = \{\mathbf{v} \; : \; \langle \pv{\vbeta}{2} - \pv{\vbeta}{0}, \mathbf{v} \rangle = -\alpha, \; \mathbf{v}[1] < 0 \}, \;\;\; \pv{\cU}{1} = \{\mathbf{v} \; : \; \langle \pv{\vbeta}{2} - \pv{\vbeta}{1}, \mathbf{v} \rangle = -\alpha, \; \mathbf{v}[1] > 0 \},
\end{equation*}
and $\pv{\mathbf{v}}{2} = [0 \;\; \zeta]^\top$, where $\alpha, \zeta > 0$.
Such a setting is possible, e.g. when $\pv{\vu}{0}_1 = [1 \;\; 0]^\top$, $\pv{\vu}{0}_2 = [0 \;\; 1]^\top$, $\sum_{t=T_0+1}^T \pv{\vu}{0}_t = [-1 \;\; 0.5]^\top$, $\sum_{t=T_0+1}^T \pv{\vu}{1}_t = [1 \;\; 0.5]^\top$, $\sum_{t=T_0+1}^T \pv{\vu}{2}_t = [0 \;\; 1]^\top$, and $\omega_{T_0 + 1} = \omega_{T_0 + 2} = \cdots = \omega_{T} = 1$.
Observe that a necessary condition for correctly intervening on units in $\pv{\cU}{0}$ is that the intervention policy should not assign intervention $d=2$ to any units with pre-treatment outcomes $\vy_{pre}$ such that $\|\vy_{pre} - \mathbf{v}\|_2 \leq \delta$, where $\mathbf{v} \in \pv{\cU}{0}$. 
This is because such $\mathbf{v}$'s could best respond and get intervention $2$ instead of their type, which is $0$. An analogous necessary condition holds for units in $\pv{\cU}{1}$. 
By~\Cref{def:ball}, any intervention policy which correctly intervenes on unit $\mathbf{v}_i$ if $\mathbf{v}_i \in \pv{\cU}{0}$ or $\mathbf{v}_i \in \pv{\cU}{1}$ must not assign intervention $d=2$ if $\Tilde{\vy}_{i,pre}$ is such that\looseness-1
\begin{equation*}
\begin{aligned}
\langle \pv{\vbeta}{2} - \pv{\vbeta}{1}, \Tilde{\vy}_{i,pre} \rangle &\in [-\alpha - \delta \|\pv{\vbeta}{2} - \pv{\vbeta}{1}\|_2, -\alpha + \delta \|\pv{\vbeta}{2} - \pv{\vbeta}{1}\|_2]\\
\text{or }
\langle \pv{\vbeta}{2} - \pv{\vbeta}{0}, \Tilde{\vy}_{i,pre} \rangle &\in [-\alpha - \delta \|\pv{\vbeta}{2} - \pv{\vbeta}{0}\|_2, -\alpha + \delta \|\pv{\vbeta}{2} - \pv{\vbeta}{0}\|_2].\\
\end{aligned}
\end{equation*}

\noindent However, if this condition is satisfied, it will be impossible to correctly intervene on unit $\mathbf{v}_i$ if $\mathbf{v}_i = \pv{\mathbf{v}}{2}$ and $\alpha, \zeta$ are small enough. 
To see this, note that in order for both 
\begin{equation*}
\begin{aligned}
\langle \pv{\vbeta}{2} - \pv{\vbeta}{1}, \Tilde{\vy}_{i,pre} \rangle &> \delta \|\pv{\vbeta}{2} - \pv{\vbeta}{1}\|_2  -\alpha\\
\text{ and } \;\; \langle \pv{\vbeta}{2} - \pv{\vbeta}{0}, \Tilde{\vy}_{i,pre} \rangle &> \delta \|\pv{\vbeta}{2} - \pv{\vbeta}{0}\|_2 -\alpha\\
\end{aligned}
\end{equation*}
to hold,
\begin{equation*}
\begin{aligned}
\delta \|\pv{\vbeta}{2} - \pv{\vbeta}{1}\|_2 - \alpha &< (\pv{\vbeta}{2} - \pv{\vbeta}{1})[2](\zeta + \delta)\\
\text{ and } \;\; \delta \|\pv{\vbeta}{2} - \pv{\vbeta}{0}\|_2 - \alpha &< (\pv{\vbeta}{2} - \pv{\vbeta}{0})[2](\zeta + \delta).
\end{aligned}
\end{equation*}
This implies that intervening perfectly on all units is not possible unless $\frac{1}{2} \zeta + \alpha > \delta(\sqrt{1.25} - 0.5)$, which does not hold for sufficiently small $\alpha$, $\zeta$.
In other words, the condition on $\alpha$, $\zeta$ implies that if the pre-intervention outcomes of units of different types are sufficiently close, intervening perfectly on these units is generally not possible.
\end{proof}

\section{Proofs from Section~\ref{sec:learning}}\label{app:learning}
\begin{theorem}\label{cor:two-learn-app}
    Suppose $d \in \{0, 1\}$.~\Cref{alg:sp-two-action} achieves out-of-sample performance
    \iffalse
    \begin{equation}\label{eq:two-learn-app}
        \frac{1}{m} \sum_{i=n+1}^{n+m} \left(\E [\pv{r}{d_{i}^A}_{i}] - \E [\pv{r}{d_{i}^*}_{i} ] \right)^2 \leq \frac{4}{m} \max_{d \in \{0, 1\}} \sum_{i=n+1}^{n+m} \left(\pv{\widehat{r}}{d}_{i} - \E [\pv{r}{d}_{i}] \right)^2
    \end{equation}
    \fi 
    \begin{equation}\label{eq:two-learn-app}
        \E [\pv{r}{d_{i}^A}_{i}] - \E [\pv{r}{d_{i}^*}_{i} ] \leq |\pv{\widehat{r}}{0}_{i} - \E [\pv{r}{0}_{i}]| + |\pv{\widehat{r}}{1}_{i} - \E [\pv{r}{1}_{i}]|
    \end{equation}
    for any test unit $i$, where $d_i^A$ is the intervention assigned to unit $i$ by~\Cref{alg:sp-two-action}, $d_i^*$ is the optimal intervention to assign to unit $i$, and $\pv{\widehat{r}}{d}_{i} := \langle \pv{\widehat{\vbeta}}{d}, \vy_{i,pre} \rangle$ is the estimated principal reward under intervention $d$.
\end{theorem}
The proof of~\Cref{cor:two-learn-app} relies on the following proposition, which shows that the interventions assigned by the intervention policy of~\Cref{alg:sp-two-action} on strategic units match the interventions assigned according to the following intervention policy on units which are always truthful. We say that a unit is truthful if they do not modify their pre-intervention outcomes.

\begin{lemma}\label{prop:not-sp-app}
    Consider the following intervention policy:
    \begin{equation}\label{eq:not-sp-app}
    %\tag*{(NOT-SP)}
        d_{i}^{B} = 
        \begin{cases}
            1 \; &\text{if } \; \; \pv{\widehat{r}}{1}_i - \pv{\widehat{r}}{0}_i > 0\\
            0 \; &\text{otherwise},
        \end{cases}
    \end{equation}
    where $\pv{\widehat{\vbeta}}{0}, \pv{\widehat{\vbeta}}{1}$ are defined as in~\Cref{alg:sp-two-action}. (Recall that $\pv{\widehat{r}}{d}_{i} := \langle \pv{\widehat{\vbeta}}{d}, \vy_{i,pre} \rangle$.)
    The intervention policy of~\Cref{alg:sp-two-action} assigns the same interventions to strategic units that intervention policy (\ref{eq:not-sp-app}) assigns to truthful units.\looseness-1
\end{lemma}
\begin{proof}
    The proof proceeds on a case-by-case basis. Fix a (strategic) unit $i \in \{n+1, \dots, n+m\}$.
    
    \noindent\textbf{Case 1}: Suppose intervention policy (\ref{eq:not-sp-app}) assigns intervention $d_{i}^B = 1$ to unit $i$. Since $d_{i}^B = 1$, $\langle \pv{\widehat{\vbeta}}{1} - \pv{\widehat{\vbeta}}{0}, \vy_{i,pre} \rangle > 0$. If Algorithm~\ref{alg:sp-two-action} assigns intervention $1$ to unit $i$ without any modification to their pre-treatment outcome, then the claim holds trivially. One valid modification is:
    %If unit $i$ modifies their pre-treatment outcomes from $\vy_{i,pre}$ to
    $$\Tilde{\vy}_{i,pre} = \vy_{i,pre} + \delta \frac{\pv{\widehat{\vbeta}}{1} - \pv{\widehat{\vbeta}}{0}}{\left \|\pv{\widehat{\vbeta}}{1} - \pv{\widehat{\vbeta}}{0} \right\|_2}.$$ Supposing unit $i$ modifies to $\Tilde{\vy}_{i,pre}$,
    $$\langle \pv{\widehat{\vbeta}}{1} - \pv{\widehat{\vbeta}}{0}, \Tilde{\vy}_{i,pre} \rangle = \langle \pv{\widehat{\vbeta}}{1} - \pv{\widehat{\vbeta}}{0}, \vy_{i,pre} \rangle + \delta \left \langle \pv{\widehat{\vbeta}}{1} - \pv{\widehat{\vbeta}}{0}, \frac{\pv{\widehat{\vbeta}}{1} - \pv{\widehat{\vbeta}}{0}}{\left \|\pv{\widehat{\vbeta}}{1} - \pv{\widehat{\vbeta}}{0} \right\|_2} \right \rangle > \delta \| \pv{\widehat{\vbeta}}{1} - \pv{\widehat{\vbeta}}{0}\|_2.$$
    Therefore, unit $i$ can receive intervention $d_{i}^A=1$ under the intervention policy of~\Cref{alg:sp-two-action}. 
        
    \noindent\textbf{Case 2:} Suppose intervention policy (\ref{eq:not-sp-app}) assigns intervention $d_{i}^B = 0$ to unit $i$, i.e., $\langle \pv{\widehat{\vbeta}}{1} - \pv{\widehat{\vbeta}}{0}, \vy_{pre} \rangle \leq 0$. To receive intervention $d_i^A = 1$ by the policy of \Cref{alg:sp-two-action}, it needs to be the case that   
    $\langle \pv{\widehat{\vbeta}}{1} - \pv{\widehat{\vbeta}}{0}, \Tilde{\vy}_{i,pre} \rangle - \delta \| \pv{\widehat{\vbeta}}{1} - \pv{\widehat{\vbeta}}{0} \|_2 > 0$.
    
    However according to~\Cref{ass:mod}, the most a unit can manipulate their pre-treatment outcomes by is $\delta$, so 
    $$\langle \pv{\widehat{\vbeta}}{1} - \pv{\widehat{\vbeta}}{0}, \Tilde{\vy}_{i,pre} \rangle \leq \left \langle \pv{\widehat{\vbeta}}{1} - \pv{\widehat{\vbeta}}{0}, \vy_{i,pre} + \frac{\delta(\pv{\widehat{\vbeta}}{1} - \pv{\widehat{\vbeta}}{0})}{\|\pv{\widehat{\vbeta}}{1} - \pv{\widehat{\vbeta}}{0}\|_2} \right \rangle \leq \delta \|\pv{\widehat{\vbeta}}{1} - \pv{\widehat{\vbeta}}{0}\|_2,$$ where for the last inequality we have used the fact that $\langle \pv{\widehat{\vbeta}}{1} - \pv{\widehat{\vbeta}}{0}, \vy_{pre} \rangle \leq 0.$ Therefore, no valid strategic modification to unit $i$'s pre-treatment outcomes exists for which the intervention policy of~\Cref{alg:sp-two-action} assigns intervention $d_i^A = 1$ to unit $i$.
\end{proof}
Since the performance of the intervention policy of~\Cref{alg:sp-two-action} on strategic units matches that of intervention policy (\ref{eq:not-sp-app}) on truthful units, we can analyze the performance of intervention policy (\ref{eq:not-sp-app}) on truthful units without any loss of generality. 
The analysis of the performance of intervention policy (\ref{eq:not-sp-app}) on truthful units completes the proof of~\Cref{cor:two-learn}.
If $d_{i}^A = d_{i}^*$, then $\E[\pv{r}{d_{i}}_{i}] - \E[\pv{r}{d_{i}^*}_{i}] = 0$. If $d_{i}^A \neq d_{i}^*$, we know that
\begin{equation*}
    \E[\pv{r}{d_{i}^*}_i] - |\pv{\widehat{r}}{d_{i}^*}_i - \E[\pv{r}{d_{i}^*}_{i}]| \leq \pv{\widehat{r}}{d_{i}^*}_{i} \leq \pv{\widehat{r}}{d_{i}^A}_{i} \leq \E[\pv{r}{d_{i}^A}_{i}] + |\pv{\widehat{r}}{d_{i}^A}_{i} - \E[\pv{r}{d_{i}^A}_{i}]|.
\end{equation*}
Therefore,
\begin{equation*}
\begin{aligned}
    \E[\pv{r}{d_{i}^A}_{i}] - \E[\pv{r}{d_{i}^*}_{i}] &\leq |\pv{\widehat{r}}{d_{i}^A}_{i} - \E[\pv{r}{d_{i}^A}_{i}]| + |\pv{\widehat{r}}{d_{i}^*}_{i} - \E[\pv{r}{d_{i}^*}_{i}]|\\
    &\leq |\pv{\widehat{r}}{0}_{i} - \E[\pv{r}{0}_{i}]| + |\pv{\widehat{r}}{1}_{i} - \E[\pv{r}{1}_{i}]|\\
\end{aligned}
\end{equation*}
\iffalse
\begin{equation*}
\begin{aligned}
    \frac{1}{m} \sum_{i=n+1}^{n+m} \left(\E[\pv{r}{d_{i}^A}_{i}] - \E[\pv{r}{d_{i}^*}_{i}] \right)^2 &\leq \frac{1}{m} \sum_{i=n+1}^{n+m} \left(|\pv{\widehat{r}}{d_{i}^A}_{i} - \E[\pv{r}{d_{i}^A}_{i}]| + |\pv{\widehat{r}}{d_{i}^*}_{i} - \E[\pv{r}{d_{i}^*}_{i}]| \right)^2\\
    &\leq \frac{1}{m} \sum_{i=n+1}^{n+m} \left(\sum_{d=0}^{1} |\pv{\widehat{r}}{d}_{i} - \E[\pv{r}{d}_{i}]| \right)^2\\
    &= \frac{1}{m} \sum_{i=n+1}^{n+m} \left( \sum_{d=0}^{1} (\pv{\widehat{r}}{d}_{i} - \E[\pv{r}{d}_{i}])^2 \right.\\
    &+ \left. \sum_{d=0}^{k-1} \sum_{d'=0, d' \neq d}^{1} |\pv{\widehat{r}}{d}_{i} - \E[\pv{r}{d}_{i}]| |\pv{\widehat{r}}{d'}_{i} - \E[\pv{r}{d'}_{i}]| \right)\\
    &\leq \frac{4}{m} \max_{d \in \range{k}_0} \sum_{i=n+1}^{n+m} \left(\pv{\widehat{r}}{d}_{i} - \E[\pv{r}{d}_{i}] \right)^2
\end{aligned}
\end{equation*}
\fi 
%
\begin{assumption}[Expected Reward Gap]\label{ass:margin-app}
Suppose that for each unit type $d$, $\pv{\vbeta}{d}, \pv{\widehat \vbeta}{d} \in [-\Bar{\beta}, \Bar{\beta}]^{T_0}$ for $\Bar{\beta} \in \mathbb{R}_+$ and there is a gap in the principal's expected reward between assigning units their type and assigning them any other intervention. Formally, for some $\alpha > 0$ (specified in~\Cref{cor:multi-app}) for each unit type $d \in \range{k}_0$: 
 $\forall v \in \pv{\cV}{d}, \;\;\; \E [\pv{r}{d}_{v} ] - \E [\pv{r}{d'}_v] > \pv{\gamma}{d,d'}$ for all $d' < d$, where $$\pv{\gamma}{d,d'} := (\sqrt{T_0} + \delta) (\| \pv{\vbeta}{d} - \pv{\widehat{\vbeta}}{d} \|_2 + \| \pv{\vbeta}{d'} - \pv{\widehat{\vbeta}}{d'} \|_2) + \delta \| \pv{\vbeta}{d} - \pv{\vbeta}{d'} \|_2 + 6 \sigma \Bar{\beta} \sqrt{2 T_0 \log(1/\alpha)},$$
and $\delta$, $\sigma$ are defined as in~\Cref{ass:mod} and~\Cref{ass:lfm-r} respectively.
\end{assumption}
The gap in~\Cref{ass:margin-app} depends on three terms: one which goes to zero as $\{\pv{\widehat{\vbeta}}{d}\}_{d=0}^{k-1} \rightarrow \{\pv{\vbeta}{d}\}_{d=0}^{k-1}$, one which is proportional to the maximum amount of modification possible in the pre-treatment period, and one which is proportional to the amount of measurement noise.
Note that under~\Cref{ass:margin-app}, separation of types (\Cref{cond:nec-suf}) holds by design (in expectation).
Intuitively, such a gap between unit rewards allows the principal to incentivize \emph{truthful} unit behavior, as it is possible to design an intervention policy such that no unit has an incentive to modify their pre-intervention outcomes. When units are truthful, linear intervention policies are optimal due to~\Cref{prop:reward}.
\begin{corollary}\label{cor:multi-app} 
    Suppose the principal's expected rewards satisfy the gap assumption (\Cref{ass:margin-app}). Then, for any $\alpha > 0 $, with probability at least $1 - \alpha$,~\Cref{alg:learning-sp} achieves out-of-sample performance\looseness-1
    \iffalse
    \begin{equation*}
        \frac{1}{m} \sum_{i=n+1}^{n+m} \left(\E[\pv{r}{d_{i}^A}_{i}] - \E[\pv{r}{d_{i}^*}_{i}] \right)^2 \leq \frac{k^2}{m} \max_{d \in \range{k}_0} \sum_{i=n+1}^{n+m} \left(\pv{\widehat{r}}{d}_{i} - \E[\pv{r}{d}_{i}] \right)^2,
    \end{equation*}
    \fi 
    \begin{equation*}
        \E[\pv{r}{d_{i}^A}_{i}] - \E[\pv{r}{d_{i}^*}_{i}] \leq \sum_{d=1}^{k} |\pv{\widehat{r}}{d}_{i} - \E[\pv{r}{d}_{i}]|,
    \end{equation*}
    for any test unit $i$, where $d_i^A$ is the intervention assigned to unit $i$ by~\Cref{alg:learning-sp}, $d_i^*$ is the optimal intervention to assign to (strategic) unit $i$, and $\pv{\widehat{r}}{d}_{i} := \langle \pv{\widehat{\vbeta}}{d}, \vy_{i,pre} \rangle$ is the estimated principal reward under intervention $d$.
\end{corollary}
The proof of~\Cref{cor:multi-app} proceeds analogously to that of~\Cref{cor:two-learn}. We begin by showing that the interventions assigned by the intervention policy of~\Cref{alg:learning-sp} on strategic units match the interventions assigned according to the intervention policy in the following lemma. However, unlike in the proof of~\Cref{cor:two-learn}, we also show that behaving truthfully in the pre-intervention period is a (weakly) dominant strategy for each unit under~\Cref{ass:margin-app}.
\begin{lemma}\label{prop:gap}
Consider the setting of~\Cref{cor:multi-app} and the following intervention policy: 
    
    \noindent Assign intervention $d_i^B = d$ to unit $i$ if
    \begin{equation}\label{eq:truthful}
    \begin{aligned}
        \pv{\widehat{r}}{d}_{i} - \pv{\widehat{r}}{d'}_{i} > 0 \; \text{ for all } \; d' < d \;\;\;
        \text {and } \;\;\; \pv{\widehat{r}}{d}_{i} - \pv{\widehat{r}}{d'}_{i} \geq 0 \; \text{ for all } \; d' > d,\\
    \end{aligned}
\end{equation}
where $\{ \pv{\widehat{\vbeta}}{d} \}_{d \in \range{k}_0}$ are defined as in~\Cref{alg:learning-sp}.
For any $\alpha>0$, the intervention policy of~\Cref{alg:learning-sp} assigns the same interventions to strategic units that intervention policy (\ref{eq:truthful}) assigns to truthful units with probability at least $1 - \alpha$.
\end{lemma}
\begin{proof}
    Suppose that the intervention policy of~\Cref{alg:learning-sp} would assign intervention $d_i^A = d$ to unit $i$ if they were truthful. Unit $i$ cannot obtain any intervention $d' > d$ under the intervention policy of~\Cref{alg:learning-sp}, due to an argument analogous to Case 2 in the proof of~\Cref{prop:not-sp-app}. 

    Next we show that if $\Tilde{\vy}_{i,pre} = \vy_{i,pre}$, the intervention policy of~\Cref{alg:learning-sp} also assigns intervention $d$ to unit $i$. 
    Consider $d' < d$.

    \begin{equation*}
        \begin{aligned}
            \langle \pv{\widehat{\vbeta}}{d} - \pv{\widehat{\vbeta}}{d'}, \Tilde{\vy}_{i,pre} \rangle - \delta \| \pv{\widehat{\vbeta}}{d} - \pv{\widehat{\vbeta}}{d'} \|_2 &\geq \E[\pv{r}{d}_i] - \E[\pv{r}{d'}_i] - 6 \sigma \Bar{\beta} \sqrt{2 T_0 \log(1/\alpha)}\\
            &-(\sqrt{T_0} + \delta)(\|\pv{\vbeta}{d} - \pv{\widehat{\vbeta}}{d}\|_2 + \|\pv{\vbeta}{d'} - \pv{\widehat{\vbeta}}{d'}\|_2)\\ 
            &- \delta \|\pv{\vbeta}{d} - \pv{\vbeta}{d'}\|_2,
        \end{aligned}
    \end{equation*}
    with probability at least $1 - \alpha$, which follows from algebraic manipulation and a Hoeffding bound. 
    Since the expected reward gap is sufficiently large, 
    $\langle \pv{\widehat{\vbeta}}{d} - \pv{\widehat{\vbeta}}{d'}, \Tilde{\vy}_{i,pre} \rangle - \delta \| \pv{\widehat{\vbeta}}{d} - \pv{\widehat{\vbeta}}{d'} \|_2 > 0$
    with probability at least $1 - \alpha$ if $\Tilde{\vy}_{i,pre} = \vy_{i,pre}$. Therefore, unit $i$ can receive intervention $d$ under the intervention policy of~\Cref{alg:learning-sp} by behaving truthfully with probability at least $1 - \alpha$.\looseness-1
\end{proof}
If $d_{i}^A = d_{i}^*$, then $\E[\pv{r}{d_{i}}_{i}] - \E[\pv{r}{d_{i}^*}_{i}] = 0$. If $d_{i}^A \neq d_{i}^*$, we know that
\begin{equation*}
    \E[\pv{r}{d_{i}^*}_i] - |\pv{\widehat{r}}{d_{i}^*}_i - \E[\pv{r}{d_{i}^*}_{i}]| \leq \pv{\widehat{r}}{d_{i}^*}_{i} \leq \pv{\widehat{r}}{d_{i}^A}_{i} \leq \E[\pv{r}{d_{i}^A}_{i}] + |\pv{\widehat{r}}{d_{i}^A}_{i} - \E[\pv{r}{d_{i}^A}_{i}]|.
\end{equation*}
Therefore,
\begin{equation*}
\begin{aligned}
    \E[\pv{r}{d_{i}^A}_{i}] - \E[\pv{r}{d_{i}^*}_{i}] &\leq |\pv{\widehat{r}}{d_{i}^A}_{i} - \E[\pv{r}{d_{i}^A}_{i}]| + |\pv{\widehat{r}}{d_{i}^*}_{i} - \E[\pv{r}{d_{i}^*}_{i}]|\\
    &\leq \sum_{d=1}^k |\pv{\widehat{r}}{d}_{i} - \E[\pv{r}{d}_{i}]|\\
\end{aligned}
\end{equation*}
%
\iffalse
\begin{equation*}
\begin{aligned}
    \frac{1}{m} \sum_{i=n+1}^{n+m} \left(\E[\pv{r}{d_{i}^A}_{i}] - \E[\pv{r}{d_{i}^*}_{i}] \right)^2 &\leq \frac{1}{m} \sum_{i=n+1}^{n+m} \left(|\pv{\widehat{r}}{d_{i}^A}_{i} - \E[\pv{r}{d_{i}^A}_{i}]| + |\pv{\widehat{r}}{d_{i}^*}_{i} - \E[\pv{r}{d_{i}^*}_{i}]| \right)^2\\
    &\leq \frac{1}{m} \sum_{i=n+1}^{n+m} \left(\sum_{d=0}^{k-1} |\pv{\widehat{r}}{d}_{i} - \E[\pv{r}{d}_{i}]| \right)^2\\
    &= \frac{1}{m} \sum_{i=n+1}^{n+m} \left( \sum_{d=0}^{k-1} (\pv{\widehat{r}}{d}_{i} - \E[\pv{r}{d}_{i}])^2 \right.\\
    &+ \left. \sum_{d=0}^{k-1} \sum_{d'=0, d' \neq d}^{k-1} |\pv{\widehat{r}}{d}_{i} - \E[\pv{r}{d}_{i}]| |\pv{\widehat{r}}{d'}_{i} - \E[\pv{r}{d'}_{i}]| \right)\\
    &\leq \frac{k^2}{m} \max_{d \in \range{k}_0} \sum_{i=n+1}^{n+m} \left(\pv{\widehat{r}}{d}_{i} - \E[\pv{r}{d}_{i}] \right)^2
\end{aligned}
\end{equation*}
\fi 
%

\end{document}